\documentclass[a4paper,USenglish,cleveref,autoref,thm-restate]{lipics-v2021}

\usepackage{mathtools}
\usepackage{xspace}
\usepackage{complexity}
\usepackage{amsfonts}
\usepackage{amssymb}
\usepackage{algorithm}
\usepackage[noend]{algorithmic}

\usepackage{booktabs}
\usepackage[textsize=small]{todonotes}
\usepackage{changepage}
\usepackage{moresize}
\usepackage{graphicx}
\usepackage{amsmath}
\usepackage{tikz}
\usetikzlibrary{arrows.meta}
\usepackage[x11names]{xcolor}

\newcommand{\tikzcurveddownarrow}[2][black]{%
  \smash[b]{%
    \tikz[baseline=-0.6ex]{%
      \draw[
        draw=#1,
        line width=1pt,
        -{Stealth[length=2mm]},
      ]
      (0,0) to[out=360, in=70] (#2,-.35);
    }%
  }%
}

\newcommand{\bZ}{\mathbb{Z}}

\newcommand{\hF}{F^1}
\newcommand{\cA}{\mathcal{A}}

\newcommand{\zero}{\ensuremath{I}\xspace}
\newcommand{\opt}{\ensuremath{\mathit{OPT}}\xspace}

\DeclareMathOperator{\OPF}{OPF}
\DeclareMathOperator{\IPF}{IPF}

\newcommand{\retrieve}{ \hspace{-0.05cm }\ensuremath{\odot}_{\hspace{-0.05cm}\textit{\ssmall P}} \hspace{-0.03cm}}
\DeclareMathOperator{\precorder}{\prec}
\DeclareMathOperator{\universal}{Universal}
\DeclareMathOperator{\runtime}{Running-Time}

\newcommand{\newstr}{Algorithm~\ref{alg:strategy_simple}-\ref{alg:subroutine}\xspace}

\graphicspath{{./figures/}}

\hideLIPIcs
\title{Instance and Universally Optimal Bounds for Imprecise Pareto Fronts}
\titlerunning{Instance and Universally Optimal Bounds for Imprecise Pareto Fronts}

\author{Sarita de Berg}{Department of Computer Science, IT University of Copenhagen, Denmark}{debe@itu.dk}{https://orcid.org/0000-0001-5555-966X}{}

\author{Nynne Maria Foldager Bække}{Technical University of Denmark, Denmark}{nynne4000@hotmail.com}{}{}

\author{Frida Astrup Eriksen}{Technical University of Denmark, Denmark}{}{}{}

\author{Ivor van der Hoog}{IT University of Copenhagen, Denmark}{ivva@itu.dk}{https://orcid.org/0009-0006-2624-0231}{}

\author{Eva Rotenberg}{IT University of Copenhagen, Denmark}{erot@itu.dk}{0000-0001-5853-7909 }{}{}

\author{Daniel Rutschmann}{Institute of Science and Technology Austria (ISTA), Austria}{rutschmann@proton.me}{ https://orcid.org/0009-0005-6838-2628}{}

\authorrunning{S. de Berg, N. M. F. Bække, F. A. Eriksen, I. van der Hoog, E. Rotenberg, D. Rutschmann}

\Copyright{Sarita de Berg, Nynne Maria Foldager Bække, Frida Astrup Eriksen, Ivor van der Hoog, Eva Rotenberg, Daniel Rutschmann}



\funding{This work was supported by the the VILLUM Foundation grant (VIL37507) ``Efficient Recomputations for Changeful Problems''.}

\nolinenumbers

\EventEditors{Mickey Mouse}
\EventNoEds{1}
\EventLongTitle{1st International Symposium on Magical Artifacts (MAGIC 2024)}
\EventShortTitle{MAGIC 2024}
\EventAcronym{MAGIC}
\EventYear{2024}
\EventDate{June 26-- July 29, 2024}
\EventLocation{Perth, Australia}
\EventLogo{}
\SeriesVolume{1}
\ArticleNo{1}
\begin{document}

\maketitle

\begin{abstract}
    In the imprecise geometry model, the input is an imprecise point set, which is a family of regions $F = (R_1, \ldots,R_n)$, where for each $R_i$ one may retrieve the true point $p_i \in R_i$. By preprocessing $F$, we can construct the output, in our case the Pareto front, on $P$ faster.

    We efficiently construct the Pareto front of an imprecise point set in the plane. Efficiency is interpreted in two ways: minimizing (i) the number of retrievals, and (ii) the computation time used to determine the set of regions that must be retrieved and to construct the Pareto front. 

    We present an algorithm to construct the Pareto front for possibly overlapping rectangles that is \emph{instance-optimal} with respect to the number of retrievals, meaning that for every fixed input $(F, P)$,  there is no algorithm that retrieves asymptotically fewer regions to compute the output. This is a strong algorithmic quality, as it means that our algorithm is competitive even to clairvoyant algorithms which know a correct guess of the output and only have to verify its correctness. In terms of algorithmic running time, instance-optimality is provably unobtainable. We instead present an algorithm which is within a $\log n$-factor of  instance-optimality. This generalizes earlier results to overlapping input regions, at only a minor cost in running time. 

    For unit squares, we present an algorithm that is not only instance-optimal in the number of retrievals, but also \emph{universally} optimal in terms of running time, meaning that for any fixed set of regions $F$, no algorithm has a better worst-case running time for all possible point sets $P$. This is the first universally optimal algorithm for overlapping planar input. Compared to previous work, this result improves the degree of overlap, the preprocessing time, the number of retrievals, and the running time.
\end{abstract}

\setcounter{page}{0}

\newpage
\section{Introduction}
\label{sec:introduction}
Computational geometry has a long tradition of studying data structures on planar point sets whose constructions reduce to sorting. 
Prominent examples include quadtrees, Euclidean minimum spanning trees, Delaunay triangulations, convex hulls, and Pareto fronts. 
Since these problems admit an $\Omega(n \log n)$ lower bound on a Real RAM, classical algorithms construct these structures in $\Theta(n \log n)$ time.
A long-standing line of research asks when this $\Omega(n \log n)$ barrier can be circumvented.
Under additional assumptions on the input, many algorithms are known that run in $o(n \log n)$ time~\cite{afshani2017instance,BuchinMulzer2011Delaunay,Chan1996,devillers2011delaunay,EppsteinEtAlCCCG2025Entropy,FranceschiniMuthukrishnanPatrascu2007,graham1972efficient,HanThorup2002,vanderhoog2025Combinatorial,vanDerHoogRustenmann2025TightUniversal,kirkpatrick1986ultimate,van2010preprocessing,loffler2013unions,loffler2010delaunay}. 
These assumptions take several forms:
One may assume that the input is already sorted~\cite{BuchinMulzer2011Delaunay,graham1972efficient, SEIDEL1985319} or partially sorted~\cite{EppsteinEtAlCCCG2025Entropy, vanDerHoogRustenmann2025TightUniversal}. 
Alternatively, one may assume integer coordinates, allowing faster-than-comparison-based sorting~\cite{FranceschiniMuthukrishnanPatrascu2007,HanThorup2002}. 
A different approach is to exploit structural properties of the input (e.g., by bounding the output complexity)~\cite{afshani2017instance,vanderhoog2025Combinatorial,kirkpatrick1986ultimate}.

We consider the case where we are given additional \emph{geometric} information about the input to speed up computation.
This assumption is formalized by the now-standard \emph{imprecise geometry} model, introduced by Held and Mitchell~\cite{held2008triangulating}.
This is a well-studied framework for obtaining faster algorithms for geometric problems~\cite{ACHARYYA2022116,deberg2025ImpreciseConvex,devillers2011delaunay, evans2011possible,ezra2013convex, held2008triangulating,van2019preprocessing,van2022preprocessing,van2010preprocessing,loffler2013unions, loffler2025preprocessing, loffler2010delaunay}, where the input is an imprecise point set. 
An imprecise point set is specified as a family of regions $F = (R_1,\ldots,R_n)$, where for each region $R_i$ one may retrieve the (unknown) true point $p_i \in R_i$. 
By preprocessing $F$, we can construct the output on $P$ faster.
In this paper, we consider the Pareto front which is the ordered sequence of maximal points.
When $F$ is a set of (overlapping) axis-aligned rectangles, we construct the Pareto front of $P$ with an optimal number of retrievals, generalizing the result of van der Hoog, Kostitsyna, L{\"o}ffler, and Speckmann~\cite{van2022preprocessing}. 
Moreover, when the regions are restricted to unit squares, we show that our algorithm achieves universally optimal running time, yielding the first universally optimal algorithm for overlapping regions in $\mathbb{R}^2$.

\subparagraph{Imprecise geometry.} An imprecise point set~\cite{held2008triangulating} is defined as a family of regions $F = (R_1, \ldots, R_n)$, where each region $R_i$ contains a unique but unknown point $p_i$. A \emph{realization} $P \sim F$ is a sequence $P = (p_1, \ldots, p_n)$ where $p_i \in R_i$. An input instance is a pair $(F, P)$. A \emph{retrieval} operation reveals the precise location of a point $p_i$, replacing $R_i$ with $p_i$.

Let $F \retrieve B$ denote the family~$F$ after retrieving all $R_i \in B$, where $B \subset F$. 
The aim of a \emph{reconstruction algorithm} is to identify a subset $B \subset F$ such that for all realizations $P_1, P_2 \sim F \retrieve B$, the algorithm’s output is the same (in our case, the ordered points on the Pareto front, see Figure~\ref{fig:problem}).
We evaluate algorithms by three criteria: the preprocessing time, the total number of retrievals, and the running time of the reconstruction algorithm.
We aim for a complexity analysis that goes beyond the (trivial) worst-case scenarios. 

\begin{figure}[h]
    \centering
    \includegraphics{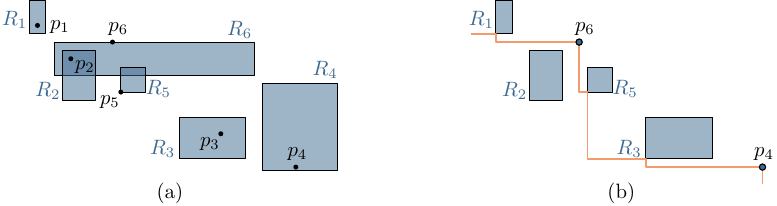}
    \caption{(a) A family of regions $F = (R_1,\ldots,R_6)$ and a sequence $P = (p_1,\ldots,p_6)$ with $P \sim F$. (b) After retrieving $R_4$ and $R_6$ the Pareto front equals $(p_1,p_6,p_5,p_3,p_4)$ for \emph{any} $P'\sim F$.
    }
    \label{fig:problem}
\end{figure}

\subsection{Worst-case analysis and beyond}

Consider an imprecise point set $(F,P)$.  
For any algorithm $\mathcal{A}$ and input $(F,P)$, we define the \emph{cost} of its execution $\texttt{cost}(\mathcal{A},F,P)$ to be either its running time $\texttt{runtime}(\mathcal{A},F,P)$ or the number of retrievals $\texttt{retrievals}(\mathcal{A},F,P)$ it performs before termination.

In \emph{worst-case} analysis, the worst case cost of an algorithm $\mathcal{A}$ is defined as
\[
\hfill
\texttt{Worst-Case}(\mathcal{A}) := \max_{\text{regions }F}\ \max_{P \sim F}\ \texttt{cost}(\mathcal{A},F,P).
\hfill
\]

\noindent
For a fixed cost measure, an algorithm $\mathcal{A}$ is \emph{worst-case optimal} if there exists a constant~$c$ such that for every algorithm $\mathcal{A}'$, $
\texttt{Worst-Case}(\mathcal{A}) \le c \cdot \texttt{Worst-Case}(\mathcal{A}')$.
If the regions in $F$ mutually overlap in some open area $A$, then $F$ conveys essentially no information about $P$: any point configuration can be scaled and translated to lie entirely inside $A$ and thus be realized within $F$.  
Consequently, for overlapping regions, the worst-case running time of a reconstruction algorithm (for all aforementioned geometric problems) is $\Omega(n \log n)$, and the worst-case number of retrievals is $\Omega(n)$.  
These bounds are matched by the naive algorithm that ignores $F$, retrieves all points, and then runs a standard construction algorithm.
To obtain meaningful worst-case bounds, much of the literature therefore restricts $F$ to consist of disjoint regions.  
Under this assumption, researchers have studied the reconstruction of Delaunay triangulations~\cite{devillers2011delaunay, held2008triangulating, van2010preprocessing,loffler2010delaunay}, convex hulls~\cite{evans2011possible, ezra2013convex}, Gabriel graphs~\cite{loffler2010delaunay}, and onion decompositions~\cite{loffler2013unions}.  
For these problems, any reconstruction algorithm must, in the worst case, perform $\Omega(n)$ retrievals.
When the regions in $F$ are pairwise disjoint unit disks~\cite{devillers2011delaunay, evans2011possible, held2008triangulating, van2010preprocessing,loffler2013unions, loffler2010delaunay}, the best known worst-case guarantees achieve $O(n \log n)$ preprocessing time, $O(n)$ retrievals and $O(n)$ reconstruction time, which is worst-case optimal.
To go beyond such worst-case guarantees, we turn to stricter notions of algorithmic analysis.

\subparagraph{Universal optimality.}
Universal optimality is based on the observation that a running-time analysis can often be strengthened by turning the outer maximum into a universal quantifier. 
For any fixed set of regions $F$, we define the \emph{universal cost} of an algorithm $\mathcal{A}$ as
\[
\hfill
\texttt{Universal}(\mathcal{A},F) := \max_{P \sim F} \texttt{cost}(\mathcal{A},F,P).
\hfill
\]
For a fixed cost measure, an algorithm $\mathcal{A}$ is \emph{universally optimal} if there exists a constant $c$ such that, for \emph{all sets of regions} $F$ and all algorithms $\mathcal{A}'$, 
$\texttt{Universal}(\mathcal{A},F) \le c \cdot \texttt{Universal}(\mathcal{A}',F)$.
For a fixed $F$, the value $\min_{\mathcal{A}} \texttt{Universal}(\mathcal{A},F)$ may range from $O(1)$ to $O(n \log n)$. E.g., there exist inputs $F$ where all realizations $P \sim F$ yield the same output and no retrieval nor reconstruction running time is needed. 
Universal optimality is considerably stronger than worst-case optimality: $\mathcal{A}$ must be efficient for all $F$, yet for each fixed $F$ it competes with algorithms tailored to $F$. 
Consequently, universally optimality is difficult to obtain.

Universal optimality has gained considerable recent attention~\cite{afshani2017instance, cardinal_sorting_2013, efremenko_et_al:LIPIcs.ITCS.2026.55, haeupler2025bidirectional, Haeupler25, hladik_et_al:LIPIcs.FORC.2025.6,  van2019preprocessing, van2022preprocessing, vanderhoog2025Combinatorial, van2025simpler, vanDerHoogRustenmann2025TightUniversal, van2024tight,kahn_entropy_1992, kahn_balancing_1984}.
For imprecise geometry, only a few results are known. 
Bruce, Hoffmann, Krizanc, and Raman~\cite{bruce2005efficient} give a polynomial-time algorithm with a universally optimal number of retrievals for very general $F$. 
L\"offler and Raichel~\cite{loffler2025preprocessing} obtain universal optimal number of retrievals and running time for convex hulls when $F$ consists of unit disks of constant overlap (ply).  Their results apply to the Pareto front also. 
De~Berg et al.~\cite{deberg2025ImpreciseConvex} achieve universally optimal retrievals for convex hulls where $F$ contains overlapping $O(1)$-gons, with reconstruction running time within an $O(\log^3 n)$ factor of universal optimality. 
Van~der~Hoog, Kostitsyna, L\"offler, and Speckmann~~\cite{van2019preprocessing} give a universally optimal algorithm if $F$ contains one-dimensional intervals and the goal is to sort $P$. Follow up work~\cite{van2022preprocessing} extends this to the two-dimensional Pareto front where $F$ consists of $n$ disjoint axis-aligned rectangles.

\subparagraph{Instance-optimality.}
Instance-optimality is an even stronger notion that replaces all maxima by universal quantifiers. 
For a fixed cost measure, instance-optimality of an algorithm $\mathcal{A}$ requires a constant $c$ such that, for all imprecise point sets $(F,P)$ and all algorithms $\mathcal{A}'$, 
$\texttt{Cost}(\mathcal{A},F,P) \le c \cdot \texttt{Cost}(\mathcal{A}',F,P)$.
This is an extremely strong requirement, as $\mathcal{A}$ must compete even with algorithms $\mathcal{A}'$ that have (deterministically) guessed the correct output and only verify their guess.
For the aforementioned geometric problems, instance-optimal running times are impossible. 
A representative example is the Pareto front with overlapping unit squares $F$: 
Consider the family of algorithms $\mathcal{A}_\pi$ indexed by permutations $\pi$, and the family of inputs $P_\pi$ whose Pareto front equals $P$ in the order $\pi$. 
Given $(F,P)$, the algorithm $\mathcal{A}_\pi$ retrieves all points and verifies in linear time whether its guess $\pi$ is correct, that is, whether $P$ forms a staircase in the order $\pi$; otherwise it constructs the Pareto front in $O(n^2)$ time. 
For each input $P_\pi$, there exists an algorithm $\mathcal{A}_i$ which it runs in linear time, yet no single algorithm can run in linear time for all $P_\pi$. 
Thus, although every $\mathcal{A}_\pi$ is individually inefficient in the worst case, they together make instance-optimal running times unattainable.

Perhaps surprisingly, instance-optimality \emph{is} achievable when the cost measure is the number of retrievals. 
We therefore call an algorithm \emph{instance-optimal} if, for some constant $c$, for all inputs, it uses at most a factor $c$ more retrievals than any other algorithm.
The retrievals performed by Bruce, Hoffmann, Krizanc, and Raman~\cite{bruce2005efficient} are not only universally optimal, they are instance-optimal. 
Similarly, the sorting algorithm of van der Hoog, Kostitsyna, L\"offler, and Speckmann~\cite{van2019preprocessing}, as well as their algorithm for computing the Pareto front of a point set when $F$ consists of disjoint axis-aligned rectangles~\cite{van2022preprocessing}, are also instance-optimal. 
Finally, de~Berg et al.~\cite{deberg2025ImpreciseConvex} give an instance-optimal algorithm for constructing the convex hull when $F$ consists of overlapping $O(1)$-gons. 
A summary is given in Table~\ref{tab:my_label}.

\begin{table}[H]
    \centering
    \begin{minipage}{\textwidth}
    \begin{tabular}{@{}llllllll@{}}
    \toprule
       Shapes 
       & Overlap & Structures & Preprocess  & Retrievals \hspace{3mm} & Reconstruction time & Source \\
    \midrule
             Unit disks  & No & many & $O(n \log n)$ & $O(n)$ & $O(n)$  & \hspace{-1.5cm}\cite{devillers2011delaunay, evans2011possible,held2008triangulating, loffler2013unions, loffler2010delaunay, van2010preprocessing} \\
          Smooth & Yes & front, hull & $O(\poly \, n) $ & $O(\texttt{instance})$ & $O( \poly \, n )$  & \cite{bruce2005efficient} \\
           Intervals & Yes &  sorting & $O(n \log n)$ &  $O(\texttt{instance})$ & $O( \texttt{universal} )$ & \cite{van2019preprocessing} \\

             \noalign{\medskip}
                    Unit disks & $k$ & hull & $O(k^3 n)$ &  $O(\texttt{universal})$ & $O( \texttt{universal} \cdot k^3)$ & \cite{loffler2025preprocessing}  \\
                           Unit disks & $k$ & hull & $O(k n \log^3 n)$ &  $O(\texttt{instance})$ & $O(  \texttt{instance} \cdot k \log^3 n)$ &  \cite{deberg2025ImpreciseConvex} \\
        
                    $k$-gons & Yes & hull & $O(kn \log^3 n)$ &  $O(\texttt{instance})$ & $O(  \texttt{instance} \cdot k \log^3 n)$ & \cite{deberg2025ImpreciseConvex} \\

             \noalign{\medskip}
      Axis rect. & No\hspace{0.5mm}\tikzcurveddownarrow[Green3]{0.05}
      & front & $O(n \log n)$ &   $O(\texttt{instance})$ & $O( \texttt{universal})$\hspace{0.5mm}\tikzcurveddownarrow[red]{0.05}& \cite{van2022preprocessing} \\   
      Axis rect. & Yes & front & $O(n \log n)$  &   $O(\texttt{instance})$ & $O( \texttt{instance} \cdot \log n)$ & Thm.~\ref{thm:rectangles} \\
          Unit disks & $k$\hspace{0.5mm}\tikzcurveddownarrow[Green3]{0.05} & front & $O(k^3 n \log n)$\hspace{0.5mm}\tikzcurveddownarrow[Green3]{0.05} &  $O(\texttt{universal})$\hspace{0.5mm}\tikzcurveddownarrow[Green3]{0.05} & $O( \texttt{universal} \cdot k^3)$\hspace{0.5mm}\tikzcurveddownarrow[Green3]{0.05} & \cite{loffler2025preprocessing}  \\
      Unit squa. & Yes & front & $O(n \log n)$  &   $O(\texttt{instance})$ & $O( \texttt{universal})$ & Thm.~\ref{thm:unit_squares} \\
      \bottomrule
    \end{tabular}
    \caption{
    Results that compute a \emph{sorting}, Pareto \emph{front}, or convex \emph{hull}.  
    For reference, our cost measures have $O( \texttt{instance}) \subseteq O(\texttt{universal}) \subseteq O(n \log n)$. Green arrows indicate that we make an improvement over the state-of-the-art. 
}
    \label{tab:my_label}
    \end{minipage}
\end{table}

\subparagraph{Our contributions.}
We present an instance-optimal Pareto front algorithm for overlapping regions. 
Its running time is within a $\log n$ factor of instance (and thus universal) optimality. 
This generalizes the result of~\cite{van2019preprocessing} to overlapping regions, at only a minor running time cost. 
For unit squares, we achieve the first universally optimal algorithm for overlapping planar input.
This can be compared with the result of L\"offler and Raichel~\cite{loffler2025preprocessing} for unit disks of ply $k$, which can also compute the Pareto front. 
In comparison, we improve the degree of overlap, the preprocessing time, the number of retrievals, and the reconstruction running time.
Table~\ref{tab:my_label} summarizes our results.

To obtain these results, we present a general instance-optimal reconstruction strategy in Section~\ref{sec:simple_strategy}. In Section~\ref{sec:overview}, we give an overview of the challenges and techniques to execute this strategy efficiently for rectangles and unit squares. We then present the algorithm to execute the strategy efficiently for rectangles in Section~\ref{sec:program_rectangles} and for unit squares in Section~\ref{sec:squares}. 

\subparagraph{Advantages of our strategy.}
Our objective is to achieve, after preprocessing, a reconstruction running time that is within a $\log n$ factor of instance-optimality. This goal is nontrivial, since this quantity depends on the unknown point set $P$ and therefore cannot be determined in preprocessing.
Existing instance-optimal reconstruction strategies are typically adaptive to the \emph{current} set of regions $F$~\cite{bruce2005efficient,van2022preprocessing}. These approaches iteratively attempt to identify, based on the information currently available about $P$, those regions in $F$ that must be retrieved by any correct algorithm. Performing this identification on-the-fly is costly: prior techniques require either polynomial time~\cite{bruce2005efficient}, time polynomial in the overlap~\cite{loffler2025preprocessing} (for convex hulls), or restrict $F$ to disjoint regions~\cite{deberg2025ImpreciseConvex,van2022preprocessing} (for convex hulls or the Pareto front).

In contrast, our strategy relies almost entirely on the \emph{original} regions. During reconstruction, the algorithm only checks whether a region had some dependency property among the original regions and whether it currently lies on the Pareto front. This design shifts much of the algorithmic intricacy to the preprocessing phase. As a result, the reconstruction phase is substantially simpler and uses logarithmic time per retrieval even for overlapping regions.

\section{Preliminaries}\label{sec:preliminaries}
The algorithmic input is an \emph{imprecise point set}, defined as a family of geometric regions $F = (R_1,\ldots,R_n)$ and a realization $P \sim F$. 
A realization $P \sim F$ is defined as a sequence $P$ of $n$ points such that $p_i \in R_i$. 
One is allowed to preprocess $F$ using polynomial time and space. Afterwards, a \emph{reconstruction algorithm} has access to the retrieve operation which replaces a region $R_i$ by its corresponding point $p_i$.  We write $F\retrieve A$ for the family that is obtained by retrieving all regions in $A \subseteq F$. When $A$ is equal to a single region $R_a$, we slightly abuse notation and use $F\retrieve R_a$ to denote the family obtained by retrieving $R_a$. We denote by $F^0$ the original family of regions before any retrievals. We denote for a family of regions $F$ by $F - A$ the family obtained by removing all regions in $A$ from $F$. 

\subparagraph{Pareto front and outer Pareto front.}
We say that a point $p$ \emph{dominates} $q$ if $p$ lies in the open top-right quadrant of $q$. 
The Pareto front $\PF(P)$ are all points of $P$ that are not dominated by some $q \in P$, ordered along the staircase bounding the area of all points $p' \in \mathbb{R}^2$ that are not dominated by some point $q \in P$. 
Note that, if $P$ does not lie in general position, then this ordering is a partial order.
We define for a family of regions $F$ the \emph{outer Pareto front} $\OPF(F)$ as all regions $R_a \in F$ such that there is a point $p \in R_a$ that is not dominated by any point $q \in R_b$ with $a \neq b$ and $R_b \in F$. 
If $F$ consists of rectangles, then the outer Pareto front corresponds to the Pareto front of the top-right corners of the regions in $F$. 
We define the \emph{inner Pareto front} $\IPF(F)$ as the area of points $p' \in \mathbb{R}^2$ such that for all $p \sim F$, $p'$ is dominated.
Note that this is an open area and that if $F$ consists of rectangles, the inner Pareto front corresponds to the area below the staircase of the Pareto front of the bottom-left corners of the regions in $F$, see Figure~\ref{fig:dependencies}. If some $R_i \in F^0$ is contained within $\IPF(F^0)$ then $p_i$ will never appear on $\PF(P)$.
Since we can detect such regions $R_i$ in $O(n \log n)$ time by constructing a Pareto front and doing intersection queries~\cite{de2008computational}, and we use at least $O(n \log n)$ preprocessing time, we henceforth assume that the input $F^0$ contains no such regions. 

\subparagraph{Restricting $F$ and $P$.} 
In this paper, we consider two scenarios. In our most general scenario, $F$ is a family of (possibly overlapping) closed axis-aligned rectangles. In the other scenario, $F$ is a family of (possibly overlapping) unit squares. 
We note that, to allow for instance-optimality, we \emph{cannot} assume that $P$ lies in general position, i.e. that all points in $P$ have distinct $x$- and $y$-coordinates. Indeed, let $F = (R_1, \ldots, R_n)$ be $n$ identical unit squares.
Consider a family of algorithms $\mathcal{A}_i$ where $\mathcal{A}_i$ retrieves $p_i$ first and a family of inputs $P_i \sim F$ where $p_i$ lies on the top-right corner of $R_i$. 
If the input is $(\mathcal{A}_i, P_i)$ then $\mathcal{A}_i$ terminates in constant time as per general position, $p_i$ dominates all remaining points in $P$.
Yet, since all regions are indistinguishable, there exists for any algorithm $\mathcal{A}$ an input $(F, P_i)$ on which $\mathcal{A}$ does $\Omega(n)$ retrievals before finding $p_i$.

\subparagraph{Dominating regions.}
Recall that by $F^0$ we denote the original input regions. If we retrieve a region $R_i$ then we replace $R_i$ by $p_i$, updating the set of regions. We denote at all times by $F$ the `current' set of regions. 
We subsequently define:

\begin{definition} Let $F$ be the current set of regions, and let $R_a \in F$. Then $R_a$ is:
\begin{itemize}
    \item \emph{always}  if for all $P \sim F$ we have that $p_a \in \PF(P)$,
    \item \emph{dominated} if for all $P \sim F$ we have that $p_a \notin \PF(P)$.
\end{itemize}
\end{definition}

\noindent
A point $p$  \emph{dominates} a region $R_a$ if $R_a$ is contained in the bottom-left open quadrant of $p$.

\subsection{Reconstruction algorithms and instance-optimality}

After preprocessing $F^0$, a \emph{reconstruction algorithm} can retrieve points to construct the Pareto front of $P$.
We are interested in the combinatorial structure, not in the points of the front itself.
For example, if the input is a family $F$ where all $P \sim F$ have the same Pareto front, then we do not wish to retrieve any points. 
We formalize this, by defining a Pareto front as a partial order and stopping the algorithm when the partial order is the same for all $P \sim F$. 

\begin{definition}
    Given $P \sim F$, let $\precorder(\PF(P))$ be the partial order on $[n]$ induced by traversing $\PF(P)$ along the corresponding staircase. I.e.,
    for $p_a, p_b \in P$, we have $a \prec b$ if and only if $p_a$ and $p_b$ both lie on $\PF(P)$ and either $p_a$ lies strictly to the left or strictly above of $p_b$.
\end{definition}
A family of regions $F$ is \emph{finished} if for any $P_1,P_2 \sim F$ we have that $\precorder(\PF(P_1))=\precorder(\PF(P_2))$.
A reconstruction algorithm retrieves some $A \subseteq F$ such that $ F \retrieve A$ is finished. We denote by $r(F,P)$  the minimum number of retrievals needed by any algorithm such that the resulting family of regions is finished.
We distinguish between two types of reconstruction algorithms: A \emph{reconstruction strategy} is analyzed only by the number of retrievals made.
 A \emph{reconstruction program} is executed on a RAM, and produces a pointer structure representing $\PF(P)$. Reconstruction programs are analyzed by both the number of retrievals and instructions.

\section{A simple instance-optimal reconstruction strategy}\label{sec:simple_strategy}

Prior instance-optimal reconstruction strategies~\cite{van2019preprocessing, van2022preprocessing, deberg2025ImpreciseConvex, bruce2005efficient} all use the same technique, introduced in~\cite{bruce2005efficient}.
    Let $F$ be the current family of regions and $A \subseteq F$. We say $A$ is a \emph{witness} if for all $P' \sim (F - A)$ the family $A \cup P'$ is not finished.
Any strategy that iteratively considers~$F$ and retrieves all regions in a  constant size witness, is instance-optimal~\cite{bruce2005efficient}.
For the Pareto front (and convex hulls) there always exists a witness of size $3$ and prior reconstruction strategies focus on finding, for the current family of regions, such a witness set. 
In this paper, we propose a simpler strategy with a different proof for instance-optimality. To this end, we define a dependency relation between regions (see Figure~\ref{fig:dependencies}):

\begin{figure}[b]
    \centering
    \includegraphics{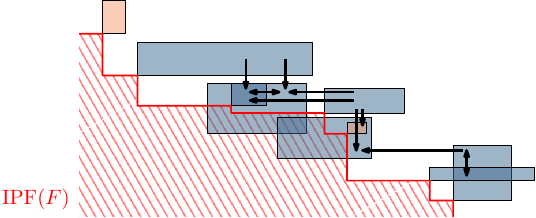}
    \caption{The dependencies between the regions in $F_0$: an outgoing arrow represents an outgoing dependency. The two orange regions have no outgoing dependencies are are thus \emph{independent} regions.}
    \label{fig:dependencies}
\end{figure}

\begin{definition}\label{def:independent}
    A region $R_a \in F_0$ has an \emph{outgoing dependency} in $F_0$ to a region $R_b \in F_0 - R_a$ if there is a point $q \in R_b$, such that all of the following conditions hold:
    \begin{enumerate}[I]
        \item\label{itm:3} the top-right corner of $R_b$ does not dominate $R_a$, and
        \item\label{itm:2} $q$ is above the inner Pareto front of $F_0$, and
        \item\label{itm:1} either $q \in R_a$, or, the top-right corner of $R_a$ dominates $q$.
    \end{enumerate}
\end{definition}

\begin{definition}
        A region $R_a \in F_0$ is \emph{independent} if it has no outgoing dependency in $F^0$.
        We denote by $I^0 \subseteq F^0$ the regions that are independent in $F^0$. We initially set $I \gets I^0$ and, whenever we retrieve $p \sim R \in I^0$, we replace the region $R$ in $I$ (but not $I^0$) by $p$.
\end{definition}

The following lemma implies that we can essentially handle the independent regions separately after dealing with the dependent regions.
\begin{lemma}\label{lem:zero_independent}
    Let $F$ be the current family of regions. Retrieving a region $R_a \in \zero$ will not cause any non-point region in $F$ to appear on the outer Pareto front.
\end{lemma}
\begin{proof}
    Suppose for contradiction that a region $R_b \in F$ appears on the outer Pareto front after retrieving $R_a$. As $R_b$ was not on the outer Pareto front before retrieving $R_a$, there must be a region in $F$ whose top-right corner strictly dominates the top-right corner $r$ of $R_b$. In particular, this region must be equal to $R_a$, as otherwise $R_b$ would still not appear on the outer Pareto front after retrieving $R_a$. Furthermore, $r$ must be above the inner Pareto front of $F^0$, as it would otherwise never appear on the outer Pareto front.
    However, the point $r$ then implies that $R_a$ has an outgoing dependency, contradiction $R_a \in \zero$.
\end{proof}

\subparagraph{A simple strategy.}
Our reconstruction strategy (Algorithm~\ref{alg:strategy_simple}) is a simple three-stage process.
In preprocessing, we compute $I^0$ and set $I = I^0$ and $F = F^0 - I^0$.
In the first stage, while there exists a region $R_a \in \OPF(F)$ that has an outgoing dependency \emph{in the original input} $F^0$, we retrieve $R_a$. Recall that $F$ is updated after each retrieval, while $F^0$ remains fixed.
In the second and third stage, we consider all regions $R_b \in \zero$ and retrieve them if it is not yet clear whether they will appear on the Pareto front (Stage 2), or if it is not yet clear \emph{where} they will appear on the Pareto front (Stage 3).  
We first prove correctness:

\begin{algorithm}[tb]
    \caption{Our simple instance-optimal reconstruction strategy.}
    \label{alg:strategy_simple}
    \begin{algorithmic}[1]
        \STATE Preprocessing: $F = F^0 - I^0$ and $\zero = I^0$
        \WHILE[Stage 1]{there exists a non-point region $R_a \in \OPF(F)$}
            \STATE Retrieve $R_a$.
        \ENDWHILE
        \FOR[Stage 2]{ all $R_a \in \zero$ }
            \IF{the interior (or the top or right facet) of $R_a$ intersects the staircase $\OPF(F)$}
                \STATE Retrieve $R_a$.
            \ENDIF
        \ENDFOR
        \FOR[Stage 3]{ all vertices $p_i \in \OPF(F)$}
            \STATE Retrieve the $R_a\in\zero$ hit by the upward ray from $p_i$, unless it hits the bottom-right.
            \STATE Retrieve the $R_b\in\zero$ hit by the rightward ray from $p_i$, unless it hits the top-left.
        \ENDFOR
    \end{algorithmic}
\end{algorithm}

\begin{lemma}\label{lem:Li_done}
    Let $F$ be the family of regions after the while loop of Algorithm~\ref{alg:strategy_simple}, then all regions in $F$ are either retrieved or dominated in $F$.
\end{lemma}
\begin{proof}
    Suppose for contradiction that there is a region $R_a \in F$ (then $R_a \not \in \zero$) that has not been retrieved and is not dominated. Then $R_a$ is not on $\OPF(F)$, as otherwise it would have been retrieved. Thus there is a region $R_b \in \OPF(F)$ such that $R_a$ is dominated by the top-right corner of $R_b$. If $R_b$ is a point region, i.e. $R_b = p_b$, then $R_a$ is dominated, a contradiction. Otherwise,  either there is an outgoing dependency from $R_b$ to $R_a$, contradicting that the while loop of Algorithm~\ref{alg:strategy_simple} has terminated, or any point $q \in R_a$ dominated by the top-right corner of $R_b$ must be below the inner Pareto front of $F^0$. As the top-right corner of $R_b$ dominates $R_a$, it follows that $R_a$ is in the inner Pareto front and thus dominated.
\end{proof}

\begin{lemma}\label{lem:finished}
    The family of regions $F^*$ at the termination of Algorithm~\ref{alg:strategy_simple} is finished.
\end{lemma}
\begin{proof}
    Let $F = F^* - \zero$. By Lemma~\ref{lem:Li_done}, $\OPF(F)$ consists of only point regions. Let $\PF^*$ denote the Pareto front of these points. We first show that every region in $F^*$ is either always or dominated, and then prove that the order along the Pareto front is fixed.

    First, consider a (retrieved) region $R_a = p_a$ on $\OPF(F)$. Assume for contradiction that $p_a$ is neither always nor dominated. Then there exists a $R_b \in \zero$ such that the top-right corner of $R_b$ dominates $p_a$ and the bottom-left corner of $R_b$ does not dominate $p_a$. This implies that the upward or rightward ray from $p_a$  hits $R_b$, and thus $R_b$ is retrieved by Algorithm~\ref{alg:strategy_simple}. There can be at most one such region in each direction, as otherwise at one of these regions would have an outgoing dependency. Thus, all regions in $F$ are either always or dominated. 

    Next, consider a region $R_c \in \zero$. When the algorithm terminates, any region in $\zero$ that intersects the staircase of $\PF^*$ in its interior, or top or right facet, has been retrieved. Thus, $R_c$ is either above or strictly below the staircase of $\PF^*$. If $R_c$ is strictly below $\PF^*$, then it is dominated. If $R_c$ is above $\PF^*$, then the point $p_c$ could be dominated only by a point in a region in $\zero$. Lemma~\ref{lem:zero_independent} thus implies that $R_c$ is always (we assumed that none of the regions in $F^0$ are dominated in $F^0$). We conclude that all regions in $F^*$ are always or dominated.

    What remains is to prove that the order of the always regions in $F^*$ along the Pareto front is fixed, i.e. for any $P_1,P_2 \sim F^*$, we have that $\precorder(\PF(P_1))=\precorder(\PF(P_2))$. Suppose there are $P_1,P_2 \sim F^*$, such that $\precorder(\PF(P_1))\neq \precorder(\PF(P_2))$. Then there are two regions $R_a$ and $R_b$ such that $a \prec b$ for $P_1$ but $a \nprec b$ for $P_2$. Note that at least one of $R_a$ and $R_b$ must be in $\zero$, as otherwise Lemma~\ref{lem:Li_done} implies that both $R_a$ and $R_b$ have been retrieved in $F^*$. Without loss of generality, assume that $R_b \in \zero$. If $R_a \in \zero$ also then the fact that there are no outgoing dependency between the two regions and both of them begin always, implies that for any $s \in R_a$ and $t \in R_b$ we have $s \prec t$, as $a \prec b$ for $P_1$. Contradicting that $a \nprec b$ for $P_2$. If $R_a \notin \zero$, then $R_a$ is retrieved. In this case, $a \prec b$ for $P_1$ but $a \nprec b$ for $P_2$ implies that $p_a \in R_b$, or the top-right corner of $R_b$ dominates $p_a$ and the bottom-left corner of $R_b$ does not dominate $p_a$. As before, we conclude that $R_b$ is retrieved, a contradiction.
\end{proof}

\subsection{Instance-optimality}\label{sec:instance_optimal_rectangles}
In this section we prove that Algorithm~\ref{alg:strategy_simple} is an instance-optimal retrieval strategy, i.e. the strategy retrieves $\Theta(r(F^0,P))$ regions.
To this end, we call a region $R \in F$ \emph{semi-independent} if $R$ has not been retrieved and $R$ had no outgoing dependencies in $F \cup \zero$ (see Definition~\ref{def:independent}). 

\begin{lemma}\label{lem:new_semi-independent}
    Retrieving a region $R_a \in (F \cup \zero)$ causes at most $4$ regions to become semi-independent.
\end{lemma}
\begin{proof}
    Consider a region $R_b \in F - R_a$ that becomes a semi-independent region in $F \retrieve R_a$. Then $R_b$ not being semi-independent before implies that there is a region $R_c$, $c \neq b$, that $R_b$ had an outgoing dependency to, and after retrieving $R_a$ one of the conditions~\ref{itm:3}, \ref{itm:2}, or~\ref{itm:1} fails. We consider two cases: either the retrieved point $p_a$ is in the inner Pareto front of $F$ or not.

    \textsf{\textbf{Case 1: \boldmath$p_a$ in $\IPF(F)$.} }
    As $p_a$ is in $\IPF(F)$, retrieving $R_a$ does not change $\IPF(F)$. Considering conditions~\ref{itm:3}, \ref{itm:2}, or~\ref{itm:1}, it must be that $c = a$ and specifically that there is an outgoing dependency from $R_b$ to $R_a$ in $F$ and this is the only outgoing dependency for $R_b$. We will show that there are at most four regions that become semi-independent by the retrieval of $R_a$ because the outgoing dependency to $R_a$ is removed.
    
    First, note that any region with an outgoing dependency to $R_a$ is by definition not  dominated by the top-right corner of $R_a$, and thus either intersects the top or right boundary of $R_a$, or is fully above or right of $R_a$. Next, we show that there is at most one region that intersects the right boundary of $R_a$ and at most one region fully right of $R_a$ that become semi-independent by retrieving $R_a$. Figure~\ref{fig:columns_empty}(c) gives an example where both of these regions exist. By symmetry, there is also at most one region that intersects the top boundary of $R_a$ and at most one region fully above $R_a$ that become semi-independent by retrieving $R_a$.

    \begin{figure}
    \centering
    \includegraphics[page=2]{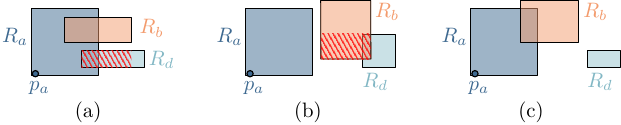}
    \caption{If the shaded red area lies in the inner Pareto front, then there is no outgoing dependency from $R_d$ to $R_a$. In (c) retrieving $R_a$ will make both $R_b$ and $R_d$ semi-independent.}
    \label{fig:columns_empty}
\end{figure}

    Suppose there are two such regions $R_b$ and $R_d$ that intersect the right boundary of $R_a$, see Figure~\ref{fig:columns_empty}(a). Without loss of generality, assume that the top of $R_b$ is at least as high as the top of $R_d$. Then $R_b$ being semi-independent after retrieving $R_a$ implies that there is no outgoing dependency from $R_b$ to $R_d$. Thus, all points in $R_d$ that are dominated by the top-right corner of $R_b$, or intersect $R_b$, (the shaded area in Figure~\ref{fig:columns_empty}(a)) must be inside the (original) inner Pareto front $\IPF(F^0)$. However, that means that all points in $R_a$ that are dominated by the top-right corner of $R_d$ must also lie inside the original inner Pareto front, contradicting that $R_d$ has an outgoing dependency to $R_a$.

    Next, suppose there are two such regions $R_b$ and $R_d$ that are fully right of $R_a$, see Figure~\ref{fig:columns_empty}(b). Again, assume without loss of generality that the top of $R_b$ is at least as high as the top of $R_d$. By definition of the inner Pareto front, no point dominated by the bottom-left corner of $R_b$ is above the inner Pareto front. Thus, if $R_d$ is fully below $R_b$, then there is no point in $R_a$ that is above the inner Pareto front that is dominated by the top-right corner of $R_d$, contradiction that $R_d$ has an outgoing dependency to $R_a$. 
    If $R_d$ is not fully below $R_b$, and the top-right corner of $R_b$ dominates $R_d$, then there being no outgoing dependency from $R_b$ to $R_d$ implies that $R_b$ is inside the original inner Pareto front. This contradicts our assumption that $F^0$ contains no dominated regions. If $R_d$ is not fully below $R_b$, and the top-right corner of $R_b$ does not dominate $R_d$, then all points in $R_b$ dominated by the top-right corner of $R_d$ must be inside the original inner Pareto front, see Figure~\ref{fig:columns_empty}(b). However, then all points in $R_a$ that are dominated the top-right corner of $R_d$ must also be inside the inner Pareto front, contradicting that $R_d$ has an outgoing dependency to $R_a$.

    \textsf{\textbf{Case 2: \boldmath$p_a$ above $\IPF(F)$.}} Because $p_a$ is above $\IPF(F)$, a region $R_b$ can only lose its outgoing dependency to another region $R_c$ with $c \neq a$. In particular, it must be that condition~\ref{itm:2} fails for $R_c$ after retrieving $R_a$. This means that any point $q \in R_c$ that is dominated by the top-right corner of $R_b$ and is above the inner Pareto front of $F$, must be dominated by $p_a$. The proof of case 1 actually implies that there are at most two regions that have their top-right corner right or above of $p_a$, and have outgoing dependencies only to points that are dominated by $p_a$. As in case 1, it thus follows there are at most four regions that become semi-independent by the retrieval of $R_a$.
\end{proof}

\begin{lemma}\label{lem:not_finished_before_independent}
    For any $R_a \in \zero$ that is retrieved by Algorithm~\ref{alg:strategy_simple}, any algorithm that does \emph{not} retrieve $R_a$ is not finished.
\end{lemma}
\begin{proof}
    We first consider Algorithm~\ref{alg:strategy_simple}.
    Let $F$ be the family of regions (excluding $\zero$) when $R_a$ is retrieved and let $P' \sim (F \cup \zero) - R_a$ be arbitrary. Note that no point in $R_a$ is dominated by any point in another region in $\zero$, as that would imply that $R_a$ was dominated in $F^0$.  We consider why $R_a$ was retrieved. 
    
    First, assume $R_a$ is retrieved because the interior or the top or bottom facet of $R_a$ intersects $\OPF(F)$.
    Let $p \in R_a$  be a point on the staircase of $ \OPF(F)$ and let $q \in R_a$ be a point below the staircase of $\OPF(F)$.
    Then setting $p_a = p$ results in $p_a$ being on $\PF(P' \cup \{p_a\})$, whereas setting $p_a = q$ results in $p_a$ not being on $\PF(P' \cup \{p_a\})$.

    Second, assume that $R_a$ is retrieved because the upwards ray from a vertex $p_i \in \OPF(F)$ hits $R_a$ not in its bottom-right corner. The case for being hit by a rightward ray is symmetric. Let $p$ be the the point where this ray hits $R_a$, let $q$ be the top-right corner of $R_a$. 
    Setting $p_a = p$ will result in both $p_a$ and $p_i$ being on $\PF(P' \cup \{p_a\})$. Setting $p_a = q$ instead, will result in $p_i \notin \PF(P' \cup \{p_a\})$.

    In both cases, we constructed two point sets $P' \cup \{q\}$ and $P' \cup \{p\}$ that only differ in the point $p_a$, but have different Pareto fronts. Thus, the algorithm $\mathcal{A}'$ that retrieves all  $R \in (F^0 - R_a)$ is not finished. 
    Any algorithm that does not retrieve $R_a$ retrieves a subset of the regions retrieved by $\mathcal{A}'$ and is therefore also not finished.  
\end{proof}

\begin{theorem}
    Algorithm~\ref{alg:strategy_simple} is an instance-optimal retrieval strategy.
\end{theorem}
\begin{proof}
    Let $(F^0, P)$ be the input.
    Fix an algorithm $\opt$ that retrieves $r(F^0, P)$ regions.
    Let $S$ denote the family of regions retrieved by $\opt$, and let $F^\opt$ be the family of regions after retrieving all of $S$, i.e. $F^\opt := F^0 \retrieve S$.
    Furthermore, let $T$ be the family of regions that become semi-independent during the execution of $\opt$.
    Lemma~\ref{lem:new_semi-independent} implies that $|T| \leq 4|S|$.
    Any region $R_a$ retrieved by our algorithm is either in $S \cup T$, or in at least one the of following cases: (1)  $R_a$ is dominated in $F^\opt$, (2) $R_a \in \zero$, or
        (3) $R_a$ has an outgoing dependency in $F^\opt$.
    Next, we show that each of these cases leads to a contradiction. Thus, any region retrieved by our strategy is in $S \cup T$, so at most $5|S| = 5\,r(F^0,P)$ regions are retrieved.

    \textsf{\textbf{Case 1.}} \normalfont Let $R_b$ be the region on the Pareto front in $F^\opt$ that dominates $R_a$. In particular, the top-right corner of (the original region) $R_b$ must dominate~$R_a$. Furthermore, $R_a$ not being dominated in $F^0$ implies that the top-right corner $q$ of $R_a$ is above $\IPF(F^0)$. The point $q$ implies an outgoing dependency from $R_b$ to $R_a$ in $F^0$, which means $R_b$ is retrieved by our strategy before $R_a$ is considered. Thus, $R_a$ is not retrieved by Algorithm~\ref{alg:strategy_simple}.

    \textsf{\textbf{Case 2.}} We apply Lemma~\ref{lem:not_finished_before_independent} to observe that \emph{any} algorithm must  retrieve $R_a$ to obtain a finished point set. So, $F^\opt$ is not finished, a contradiction.

    \textsf{\textbf{Case 3.}} By $F^\opt$ being finished and $R_a$ not being dominated, it must be that $R_a$ is always and thus on $\OPF(F^\opt)$. As $R_a$ has an outgoing dependency in $F^\opt$, there must be a region $R_b$ with a point $q \in R_b$ such that $q$ is above $\IPF(F^\opt)$, and $q \in R_a$ or the top-right corner of $R_a$ dominates $q$. In both cases, setting $p_b = q$ and $p_a$ as the top-right or bottom-left corner of $R_a$ results in two distinct Pareto fronts. We conclude that $F^\opt$ is not finished, a contradiction.
\end{proof}

\section{Overview of the reconstruction programs}\label{sec:overview}
In the remainder, we present two reconstruction programs that implement our reconstruction strategy (Algorithm~\ref{alg:strategy_simple}). We emphasize the techniques used and the challenges that we face.

\subparagraph{Preprocessing.}
The reconstruction strategy requires preprocessing the original input $F^0$ by classifying regions into those that are independent and those that have at least one outgoing dependency. The full set of dependencies has quadratic size. Fortunately, it suffices to determine whether each region has \emph{any} outgoing dependency.
A region $R$ has an outgoing dependency to a region $R'$ if certain geometric conditions are satisfied (Definition~\ref{def:independent}) \emph{and} the top-right corner of $R'$ does not dominate that of $R$. The main difficulty lies in identifying these geometric conditions efficiently for \emph{only} those regions $R'$ whose top-right corner does not dominate. Such filtering could be achieved using standard techniques at the cost of additional logarithmic factors in the preprocessing time.
To avoid any compromise in preprocessing time, we construct for each region two representative points: the leftmost point of its top facet, and the bottommost point of its right facet that lies above the inner Pareto front. We show that a region $R$ has \emph{an} outgoing dependency if and only if its top-right corner dominates at least one of these points, or is in a slightly degenerate case where one of its facets intersects a facet of another region. This characterization reduces the preprocessing task to performing $O(n)$ orthogonal range counting queries, resulting in $O(n\log n)$ preprocessing time.

\subparagraph{Reconstructing from rectangles.}
When $F$ is a family of (overlapping) rectangles, our algorithm achieves a running time of $O(r(F^0, P)\log n)$. The main challenge is to repeatedly identify a non-retrieved region $R$ on $\OPF(F)$ in amortized logarithmic time.
Our approach preprocesses $F = F^0 - I^0$ into a \emph{layer decomposition} by iteratively removing $\OPF(F)$ from~$F$ until no regions remain. During reconstruction, we retrieve  regions that are not dominated in the current family of regions $F$, layer-by-layer. The key observation is that a point retrieved from layer $i$ can never dominate any region in the same layer. Consequently, all regions that lie on $\OPF(F)$ at the beginning of processing layer $i$ remain on $\OPF(F)$ throughout that layer.
After completing layer $i$, we must remove from the remaining layers all regions that are dominated by the newly retrieved points. By carefully propagating information between layers, we are able to perform these removals implicitly, spending only logarithmic time per retrieved point. As a result, Stage~1 of the reconstruction runs in $O(r(F^0, P)\log n)$ time.
The remaining stages are handled by merging the Pareto front $\Phi(F)$ of the retrieved points with the Pareto front of the independent regions, which we do in $O(|\Phi(F)|\log n)$ time.

\subparagraph{Reconstructing from squares.}
When $F$ consists of unit squares, our goal is to achieve universal optimality. This is a very strong performance requirement: in particular, even spending $O(\log n)$ time per retrieved point is already too expensive.
To reason about universal optimality, fix an input $F^0$ and consider the number of combinatorially distinct Pareto fronts $\PF(P)$ that can arise over all realizations $P \sim F^0$. Let $U(F^0)$ denote this number. Universal optimality restricts us to a total running time of $\Theta( \log (U(F^0)))$. Equivalently, every unit of running time must be charged to a set of realizations $\{P_j\} \sim F^0$ that collectively generate sufficiently many distinct Pareto fronts $\PF(P_j)$ to pay for that cost.
Establishing this running time requires a highly technical charging and counting scheme that we detail in Section~\ref{sec:squares}.

\section{An efficient reconstruction program for rectangles}\label{sec:program_rectangles}
In this section we present a reconstruction program that executes Algorithm~\ref{alg:strategy_simple} in $O(\log n)$ time per retrieval using $O(n \log n)$ preprocessing time. We first adapt Algorithm~\ref{alg:strategy_simple} slightly, such that we have more control over the order the dependent regions, i.e. the order that regions in $F^0 - I^0$, are retrieved in. To this end, we partition $F^0 - I^0$ into several \emph{layers}.

\begin{definition}[Figure~\ref{fig:layers}(a)]
    The family $F = F^0 - I^0$ is partitioned into \emph{layers} $F_1,\dots, F_k$:
    \begin{equation*}
        F_i := \begin{cases}
            \{R_a \mid R_a \in \OPF(F^0)\} &\text{if $i = 1$}\\
            \left\{R_a \mid R_a \in \OPF\left(F^0 - \bigcup_{j=1}^{i-1} F_{j}\right)\right\} &\text{if $i > 1$ and this family is non-empty}.
        \end{cases} 
    \end{equation*}    
\end{definition}

This layer decomposition is \emph{static}, it is solely defined based on the original regions $(F^0 - I^0)$ and remains fixed during the reconstruction phase.
The first layer contains all regions of $F^0$ that are on $\OPF(F^0)$.
The subsequent layers contain all regions that are on the outer Pareto front after removing all regions in earlier layers.

\begin{observation}\label{obs:layer_order}
    A region $R_a$ in layer $F_i$ is never dominated by a point $p \in R_b$ with $R_b \in F_j$, where $j \geq i$.
\end{observation}

\subsection{An adapted reconstruction strategy}\label{sec:strategy_rectangles}
In the adapted reconstruction strategy, we will go through the layers one by one, and retrieve some regions from each layer along the way.
However, for the current family of regions $F$, we might no longer be interested in some of the regions in a layer.
We therefore define the following subfamily of each layer $F_i$ called the \emph{active layer}.
\begin{definition}[Figure~\ref{fig:layers}(b)]
    For the current family of regions $F$, an \emph{active layer} $L_i(F)$ is the subfamily of regions in $F_i$, where $R_a \in L_i(F)$ if and only if the following conditions hold:
    \begin{itemize}
        \item $R_a \in F_i$, and
        \item $R_a$ is not dominated in $F$, and
        \item $R_a$ is not retrieved.
    \end{itemize}
\end{definition}

\begin{figure}
    \centering
    \includegraphics{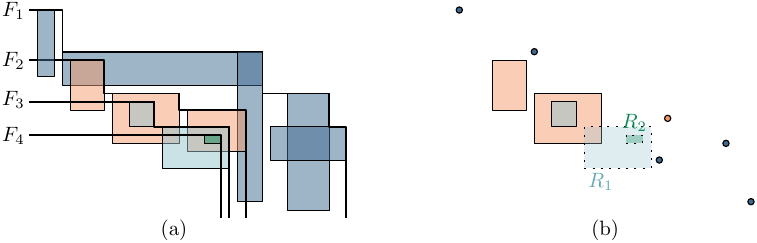}
    \caption{(a) The partition of the regions in $F_0$ into layers $F_1,F_2,F_3,F_4$. (b) After retrieving some regions, the active layer $F_i(F)$, $i \in [4]$, consists of the regions with solid boundaries.  $R_1$ and $R_2$ are not in $L_3(F)$ and $L_4(F)$ because they are dominated, point-regions are never in $R_i(F)$.\vspace{-0.3cm}}
    \label{fig:layers}
\end{figure}

To obtain an adapted reconstruction strategy, we replace Stage 1 of Algorithm~\ref{alg:strategy_simple} by the subroutine Algorithm~\ref{alg:subroutine}. Specifically, instead of retrieving any region on $\OPF(F)$ with an outgoing dependency in $F^0$, we first retrieve the regions in the active layer $L_1(F)$, then the regions in $L_2(F)$, etc. We denote this adapted strategy by \newstr.

\begin{algorithm}[h]
    \caption{A subroutine for Algorithm~\ref{alg:strategy_simple} to retrieve regions on $\OPF(F)$.}
    \label{alg:subroutine}
    \begin{algorithmic}
        \FOR{ $i = 1$ up to $k$ }
            \WHILE{$L_i(F) \neq \emptyset$}
                \STATE Retrieve any region $R_a \in L_i(F)$.
            \ENDWHILE
        \ENDFOR
    \end{algorithmic}
\end{algorithm}

\noindent
When processing an active layer $L_i(F)$, all regions in $L_i(F)$ are on the outer Pareto front:

\begin{lemma}\label{lem:layer_subset_OPF}
    When \newstr considers an active layer $L_i(F)$, then $L_i(F) \subseteq \OPF(F)$.
\end{lemma}
\begin{proof}
    Let $R_a \in L_i(F)$ with top-right corner $r$. We have that $L_j(F) = \emptyset$ for all $j < i$, and thus every region in $\bigcup_{j = 0}^{i-1} F_j$ is either retrieved or dominated. It follows that $R_a$ is not dominated by the top-right corner of a region in $\bigcup_{j = 0}^{i-1} F_j$. By Observation~\ref{obs:layer_order}, the point $r$ is also not dominated by the top-right corner of any region in $\bigcup_{j = i}^{k} F_j$. It follows that $r$ is on the staircase of $\OPF(F)$, and thus $R_a \in \OPF(F)$.
\end{proof}
Together with Observation~\ref{obs:layer_order} this implies that \newstr executes Algorithm~\ref{alg:strategy_simple}.

\begin{corollary}
    \label{cor:newstr}
    \newstr is an instance-optimal retrieval strategy.
\end{corollary}

\subsection{An efficient reconstruction program}

The first step in the preprocessing is to compute the the set $I^0$ of independent regions.
Our goal is to do this in $O(n \log n)$ time and $O(n)$ space.
To this end, we derive a more convenient characterization of independent regions.
In \cref{def:independent}, there are infinitely many choices for the point $q \in R_b$.
We show that in most cases only two of these choices are relevant:
For a region $R \in F^0$, let $\ell(R)$ be the leftmost point on the top facet of $R$ that lies above $\IPF(F^0)$,
and let $r(R)$ be the bottommost point on the right facet of $R$ that lies above $\IPF(F^0)$, see Figure~\ref{fig:finding_independents}(a).

\begin{lemma} \label{lem:depend_facet_points}
    $R_a$ has an outgoing dependency to $R_b$ if and only if there is a point $q$ on the top or right facet of $R_b$ that satisfies \cref{def:independent}.
\end{lemma}
\begin{proof} 
    The ``if'' direction is trivial, so let's prove the ``only if'' direction.
    Let $R_a$ have an outgoing dependency to $R_b$, certified by a point $\tilde{q}$ that satisfies \cref{def:independent}.
    Consider the closed line segment $S$ from $\tilde{q}$ to the top-right corner of $R_a$.
    Then all points on $S$ satisfy conditions~\ref{itm:2} and~\ref{itm:1} of \cref{def:independent}.
    Let $q$ be the top-rightmost point of $S \cap R_b$.
    If $q$ is on the top or right facet of $R_b$, this $q$ shows the ``only if'' direction.
    Otherwise, $q$ is the top-right corner of $R_a$. But then, the top-right corner of $R_b$ dominates $R_a$ --  a contradiction to condition~\ref{itm:3} of \cref{def:independent}.
\end{proof}

\begin{corollary}[Figure~\ref{fig:finding_independents}(b)] \label{cor:fast_depend}
    $R_a$ has an outgoing dependency to $R_b$ if and only if either
    \begin{enumerate}[(1)]
        \item there is a point $q \in \{\ell(R_b), r(R_b)\}$ such that the top-right corner of $R_a$ dominates $q$; or
        \item the top or right facet of $R_a$ intersects a top or right facet of $R_b$ at a point above $\IPF(F^0)$.
    \end{enumerate}
\end{corollary}
\begin{proof}
    We first show the ``only if'' direction.
    By \cref{lem:depend_facet_points}, there is a point $q$ on the, without loss of generality, top facet of $R_b$ that satisfies \cref{def:independent}.
    Let $q$ be the leftmost such point.
    Consider condition~\ref{itm:1}: If $q$ lies on the top or right facet of $R_a$, then case (2) applies.
    Otherwise $q \in \{\ell(R_b), r(R_b)\}$, and case (1) applies.

    For the ``if'' direction, in both cases, we first verify condition \ref{itm:3} of \cref{def:independent}.
    Indeed, in neither case can the top-right corner of $R_b$ dominate $R_a$.
    Then, in case (1), the point $q$ satisfies conditions \ref{itm:2} and \ref{itm:1}.
    In case (2), any intersection point satisfies these conditions.
\end{proof}

\begin{figure}
    \centering
    \includegraphics[page=2]{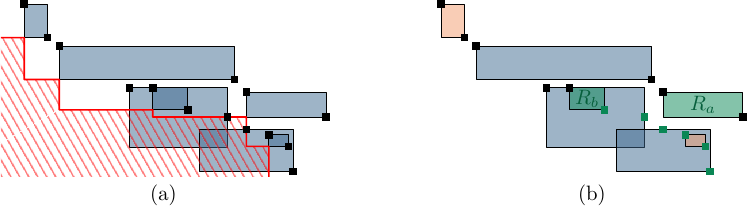}
    \caption{(a) The squares indicate the points $\ell(R)$ and $r(R)$. (b) $R_a$ has an outgoing dependency in case (1) because of the green points, while $R_b$ has an outgoing dependency in case (2). The orange regions are independent.}
    \label{fig:finding_independents}
\end{figure}

What remains is to turn \cref{cor:fast_depend} into an algorithm for finding the independent regions.
We first compute $\IPF(F^0)$ and build an orthogonal range query data structure over all $\ell(R), r(R)$ in $O(n \log n)$ time and $O(n)$ space~\cite{Chazelle88rangesearching}. 
For every region $R_a$, to check whether case (1) of \cref{cor:fast_depend} applies, we query the open area dominated by the top-right corner of $R_a$.
Next, we think of the top and right facet of each region as closed rectangles of width $0$, 
and build a rectangle-rectangle intersection counting data structure~\cite{Chazelle88rangesearching} over these.
For every region $R_a$, we query the top facet, and query the right facet. In total, there are always at least $4$ intersections (from the region $R_a$ itself).
Case (2) of \cref{cor:fast_depend} applies if and only if there are more than four.
The independent regions are precisely those for which neither case applies.
This shows:

\begin{lemma} \label{lem:compute_independent}
    Given $F^0$, we can compute $I^0$ in $O(n \log n)$ time and $O(n)$ space.
\end{lemma}

\subparagraph{Preprocessing phase.}

We first invoke \cref{lem:compute_independent} to find $I^0$ in $O(n \log n)$ time and $O(n)$ space.
We then set $F = F^0 - I^0$ and construct $\OPF(\zero)$ in $O(n \log n)$ time.
We also store $I^0$ in an orthogonal ray shooting data structure to prepare for Stages 2 and 3 of Algorithm~\ref{alg:strategy_simple}. 

Next, we partition $F$ into a layer decomposition~\cite{Chazelle2006ConvexLayers} which finds the partition $\{ F_i \}$ and constructs $\OPF(F_i)$ for each index $i$. 
We augment it with a planar point location data structure that for any query point $q$ returns the layer $F_i$ such that the downward vertical ray from $q$ hits $\OPF(F_i)$.  
This allows us to efficiently determine between which layers a retrieved point lies. The layer decomposition, and a corresponding point location data structure, can be computed in $O(n \log n)$ time and $O(n)$ space by a straightforward adaption of Chazelle's algorithm to compute the convex layers of a planar point set~\cite{Chazelle2006ConvexLayers} on the top-right corners of the regions.

\begin{definition}
    We store for each layer $i$ a tree $T_i$ that stores the regions in $\OPF(F_i)$ in-order as a leaf-based leaf-linked balanced binary tree.
    Moreover, we store a list of \emph{some} maximal intervals along $\OPF(F_i)$ that are dominated by points in $F_j$ for $j < i$. 
\end{definition}

This tree will later allow us to efficiently disregard all maximal intervals of consecutive leafs that are dominated in $F$. Specifically, it will allow us to maintain a doubly linked list storing all regions in $\OPF(F_i)$ that are not dominated in $F$.

\subparagraph{Reconstruction phase.}
As we execute the subroutine in Algorithm~\ref{alg:subroutine}, we retrieve all regions in the active layers, layer by layer.
We maintain the following invariant:

\begin{quote}
    When we process the active layer $L_i(F)$, we store \emph{all} maximal intervals of points on $\OPF(F_i)$ dominated by points in $F_j$ for $j < i$, and a doubly linked list of $L_i(F)$. 
\end{quote}

Note that this requirement only has to hold when processing $L_i(F)$. 
We achieve this invariant as follows. 
When we process an active layer $L_i(F)$, we assume the invariant to be true.
Note that this is trivially true when we start processing $L_1(F)$.
Moreover, observe that if it is true at the start of processing $L_i(F)$ then it remains true throughout as no region in $L_i(F)$ can be dominated by a point that is retrieved from the same layer. 

Our analysis relies on the following lemma, which intuitively states that two consecutive maximal intervals of dominated points along $\OPF(F_i)$ can be merged into one, \emph{if} there does not exists a region $R \in F_i$ between them:

\begin{lemma}\label{lem:safe_merge}
    Let $I_{p_a}$ and $I_{p_b}$ be two intervals along $\OPF(F_i)$ that are dominated by points $p_a$ and $p_b$, respectively.  Then either $I_{p_a} \cap I_{p_b} \neq \emptyset$, or there is region $R_a \in F_i$ in $\OPF(F_i)$ between $I_{p_a}$ and $I_{p_b}$.
\end{lemma}
\begin{proof}
    Suppose that $I_{p_a} \cap I_{p_b} = \emptyset$, see Figure~\ref{fig:merge_intervals}. Then there must be a top-right corner $q$ of the staircase $\OPF(F_i)$ between the two intervals. However, this implies that there is a region $R_a \in \OPF(F_i)$ who's top-right corner is $q$. It follows that there is a region in $\OPF(F_i)$ between $I_{p_a}$ and $I_{p_b}$.
\end{proof}

\begin{figure}
    \centering
    \includegraphics[page=2]{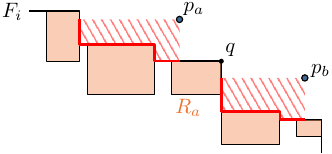}
    \caption{The two intervals $I_{p_a}$ and $I_{p_b}$ not overlapping implies that there is a top-right corner $q \in \OPF(F_i)$ between them.}
    \label{fig:merge_intervals}
\end{figure}

\emph{Our retrieval update.}
Whenever our algorithm retrieves a region $R_a \in L_i(F)$, it will update the list of intervals on $\OPF(F_{j+1})$ for exactly \emph{one} $j$. A merging strategy then ensures that the invariant holds for the active layer as we process it. 
Specifically, when we retrieve a region $R_a \in L_i(F)$ we perform a point location with $p_a$ on the layer decomposition. If $p_a$ is in the inner Pareto front, we discard it. Otherwise, suppose that $p_a$ lies between $\OPF(F_j)$ and $\OPF(F_{j+1})$.
The point $p_a$ then dominates a contiguous interval along $\OPF(F_{j+1})$ (which corresponds to a contiguous interval of leaves in $T_{j+1}$).
We search along $T_{j+1}$ to identify this interval in $O(\log n)$ time. 
We then insert this interval $I_{p_a}$ into the ordered list of maximally dominated intervals, iteratively merging $I_{p_a}$ with any interval $I_{p_b}$ in the list that is intersected by $I_{p_a}$. 
A merge between two intervals replaces both by their union. 
Note that this merge takes amortized constant time since there are at most $O(1)$ such intervals $I_{p_b}$ that are not contained in $I_{p_a}$ and all intervals contained in $I_{p_a}$ are permanently removed. 

\emph{Maintaining our invariant.}
When we are finished processing the active layer $L_i(F)$, we must ensure that our invariant holds for $L_{i+1}(F)$. We can now rely on Lemma~\ref{lem:safe_merge}.
Per our invariant (and the fact that retrieving regions in $F_i$ does not change this invariant), we stored the set of maximal intervals along $\OPF(F_i)$ that are dominated by a point in $F_j$ for $j < i$. 
Per our merging strategy, we merged these intervals upon inserting them such that between any two consecutive intervals there is a (non-dominated) region $R$ in $F_i$ (which is thus a region of $L_i(F)$).
Every such region $R$ is retrieved by our retrieval strategy, and we charge $O(\log n)$ to every such interval. 
Each interval along $\OPF(F_i)$ corresponds to some fictive point $q$ that dominates the interval and we search the balanced tree storing  $\OPF(F_{i+1})$  in $O(\log n)$ time to find the maximal interval along  $\OPF(F_{i+1})$ that is dominated by $q$. 
With this total charge, we search for every such fictive $q$ for the maximal interval along  $\OPF(F_{i+1})$ that is dominated by $q$. 
We merge these new intervals with the intervals already stored along  $\OPF(F_{i+1})$  in the exact same manner as above at amortized constant overhead. 
Every point in $F_j$ with $j < i$  that dominates any point along $\OPF(F_{i+1})$ must have had an interval either along $\OPF(F_{i+1})$  or $\OPF(F_{i})$ by our invariant. 
It follows that this merging procedure maintains our invariant for $\OPF(F_{i+1})$ and charges logarithmic time to each retrieved region in $L_i(F)$. 
Via the instance-optimality of our retrieval strategy, we spend $O(r(F^0, P) \log n)$ total time on these retrievals.

\subparagraph{The independent regions.} 
For each region $R_a$ that we retrieve in Stage 1, we record the point $p_a$ that is retrieved. Let $F$ be the current family of regions then we denote by $\Phi(F)$ the Pareto front of all points that we retrieved, which we can maintain in $O(\log n)$ time per retrieved point.  
All regions in the layer decomposition have an outgoing dependency in $F^0$. So,  if there exists at least one region on $\OPF(F)$ that is not a point region then \newstr retrieves this region. As a result,  $\OPF(F) = \Phi(F)$ by the end of Stage 1. 

We then execute Stage 2 and Stage 3 using orthogonal range searching queries on $I^0$.
In particular, we note that for each edge $e$ of the staircase of $\Phi(F)$, only one region of $I$ can intersect $e$ in its interior, or bottom or left facet (if there are two such regions then one of the two has an outgoing dependency to the other). However, there may be many regions who's top or right facet intersect $e$, see figure~\ref{fig:program_rectangles}.
We thus shoot a ray from each vertex of $\Phi(F)$ in each cardinal direction, charging to the retrieved vertex, spending $O(r(F^0, P) \log n)$ time. If such a ray hits a region that intersects the corresponding edge only in its top (or right) facet, then we continue to the next region along $\OPF(I^0)$ as long as this next region also intersects the same edge only in its top (or right) facet.
As a consequence, we obtain a superset of all regions retrieved in Stage 2 and 3. For each of these regions, we check if they meet the requirement for retrieval at constant overhead and retrieve them if needed. It follows that we 
execute \newstr after $O(n \log n)$ preprocessing in $O(r(F^0, P) \log n)$ time to obtain a set of regions that is \emph{finished}.

Finally, we note that we can slightly strengthen the result.
We can not only find a set of regions that is \emph{finished}, we can also output a pointer to a balanced tree that stores the Pareto front in order. 
Specifically, we 
can merge the fronts $\OPF(I)$ and $\Phi(F)$ into a staircase of regions that has a 1:1 correspondence to $\PF(P)$.
As we retrieve regions from $I$ the left-to-right order along $\OPF(I)$ does not change. 
We can therefore use any balanced tree over $\OPF(I^0)$ as a balanced tree over $\OPF(I)$. 
Our next key observation is that each edge of a Pareto front intersects another Pareto front exactly once. For each edge of the staircase of $\Phi(F)$, we can binary search in $O(\log n)$ time to find the edge of $\OPF(I)$ that is intersected by $\Phi(F)$ (if any).
Given $O(|\Phi(F)|) \subset O(r(F^0, P))$ intersection points, we compute the joint maximum of $\OPF(I)$ and $\Phi(F)$ in $O(r(F^0, P))$ time which corresponds to $\PF(P)$.

\begin{figure}
    \centering
    \includegraphics[page=2]{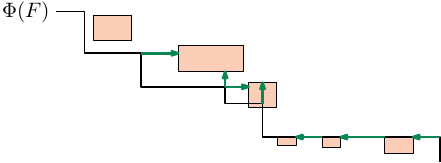}
    \caption{For each region in $I^0$, the arrows indicate the ray shooting queries that imply the region should be retrieved.}
    \label{fig:program_rectangles}
\end{figure}

\begin{theorem}\label{thm:rectangles}
    Let $F$ be a family of $n$ axis-aligned rectangles. We can process $F$ in $O(n \log n)$ time and $O(n)$ space, such that given $P\sim F$ we reconstruct the Pareto front $\PF(P)$ in $O(r(F,P) \log n)$ time.
\end{theorem}

\section{A universally optimal reconstruction program for unit squares} 
\label{sec:squares}

Suppose now that the regions in $F^0$ are all unit squares.
Consider Stage~3 of \cref{alg:strategy_simple}:
we retrieve $R_a \in I^0$ if the upward ray from a point $p_i \in \OPF(F)$ hits $R_a$. If $R_a$ does not intersect the staircase of $\OPF(F)$, then
the top-right corner of the unit square $R_i \in F^0$ can not dominate $R_a$. However, this implies that~$R_a$ has an outgoing dependency to $R_i$, contradicting the fact that $R_a \in I_0$.
Thus, for unit squares, Stage~3 does not do any retrievals, so \cref{alg:strategy_simple} is equivalent to the two-stage \cref{alg:strategy_unit_square}.
Let $F^1 = F^0 - I^0$ and let $P^1$ denote points of $P$ corresponding to $F^1$, i.e. $F^1$ is the static family of regions in $F^0$ that are not independent.

\begin{algorithm}[tb]
    \caption{The unit square reconstruction strategy.}
    \label{alg:strategy_unit_square}
    \begin{algorithmic}[1]
        \STATE Preprocessing: $F = F^0 - I^0$
        \WHILE[Stage 1]{there exists a non-point region $R \in \OPF(F)$}
            \STATE Retrieve $R$.
        \ENDWHILE
        \FOR[Stage 2]{ all $R \in \zero$ }
            \IF{the interior (or the top or right facet) of $R$ intersects the staircase $\OPF(F)$}
                \STATE Retrieve $R$.
            \ENDIF
        \ENDFOR
    \end{algorithmic}
\end{algorithm}

\subparagraph{Universal Lower Bounds.} Recall that the universal running time of an algorithm $\cA$ is $\universal(\cA, F^0) = \max_{P \sim F^0} \runtime(\cA, F^0, P)$. We say $f(F^0)$ is a \emph{universal lower bound} if,
for every algorithm $\cA$, we have $\universal(\cA, F^0) \ge f(F^0)$.
Or, equivalently, if for every algorithm $\cA$ there is a $P \sim F^0$ such that $\runtime(\cA, P) \ge f(F^0)$.

\subparagraph{Stage 2.}
We first design an efficient implementation of Stage 2,
that runs in $O(|F^1|)$ time. The following lemma then implies that this is universally optimal.

\begin{lemma} \label{lem:linear_lower_bound}
    There is a universal lower bound of $\Omega(|F^1|)$. 
\end{lemma}
\begin{proof}
    Given a family of regions $F^0$, construct $P$ by selecting the bottom-left corner for every region in $F^0$. Recall that $F^0$ does not contain any dominated regions.
    Running \cref{alg:strategy_unit_square} on $(F^0,P)$ will thus retrieve every region in $F^1$ during Stage 1.
    Since \cref{alg:strategy_unit_square} does an instance-optimal number of retrievals,
    any algorithm has to do $\Omega(|\hF|)$ retrievals,
    and hence spend $\Omega(|\hF|)$ time.
\end{proof}

The regions in $I^0$ are pairwise disjoint, and have no dependencies among each other.
Thus, for every region $R \in F^1$ there are only constantly many regions in $I^0$ that intersect the area dominated by the top-right corner of $R$.
In Stage 2, these are the only regions in $I^0$ that could intersect the staircase step of $\OPF(F)$ of a point $p \sim R$.
During preprocessing, we store $I^0$ in an array, sorted by decreasing y-coordinate.
For each region $R \in F^1$, we store a pointer to the predecessor of the top-right corner of $R$ in this array.
In Stage 2, we iterate over the points on $p \in \OPF(F)$.
For each such point $p$, we use the predecessor pointers to check the constantly many regions in $I^0$ that could intersect the staircase step at $p$.
This takes $O(|\OPF(F)|) \subseteq O(|F^1|)$ time.

\subparagraph{Stage 1.}
What remains is to design a universally optimal implementation of Stage 1.
We use a sweep line algorithm for this.

\begin{definition} \label{def:sweep_order}
    The \emph{sweep order} is the partial order defined on the (infinite) set of all unit squares and all points in $\mathbb{R}^2$ by, in decreasing order of priority:
    (1) sorting by decreasing y-coordinate of the point or top-right corner of the square, (2) placing unit squares before points, (3) sorting points by increasing x-coordinate. 
\end{definition}

We design a sweep line algorithm that constructs $\PF(P)$ from top-left to bottom-right
by keeping track of the last known point $p^* \in \PF(P)$:
Consider the regions $R \in F$ in sweep order (\cref{def:sweep_order}).
The region $R$ lies on $\OPF(F)$ if and only if $R$ is not dominated by $p^*$.
In this case, we retrieve $R$ if it is not a point, or set $p^* \gets R$ otherwise.
The resulting \cref{alg:slow_sweep_line} implements Stage 1 of \cref{alg:strategy_unit_square} in $O(n \log n)$ time.
In the next sections, we will speed up this sweep line algorithm.

\begin{algorithm}[tb]
    \caption{The generic sweep line algorithm.}
    \label{alg:slow_sweep_line}
    \begin{algorithmic}
        \STATE $p^* \gets (-\infty, \infty)$
        \STATE $T \gets $ binary search tree of $F^1$ in sweep order
        \STATE $L \gets$ empty linked list
        \FOR{ $R \in T$ in sweep order }
            \IF{$R$ is not dominated by $p^*$}
                \IF{$R$ is a point}
                    \STATE $p^* \gets R$
                    \STATE Append $p^*$ to $L$
                \ELSE
                    \STATE $p \gets$ retrieve $R$
                    \STATE Insert $p$ into $T$
                \ENDIF
            \ENDIF
        \ENDFOR
        \RETURN $L$
    \end{algorithmic}
\end{algorithm}

\subsection{Classifying regions}
To speed up the sweep line algorithm, 
consider an axis aligned unit grid, which splits the plane into (open) \emph{cells}.
We denote each cell by its bottom-left corner $(i, j) \in \bZ^2$.
By shifting $F^0$ (and $P$) slightly, we may assume that no corner of a region in $F^0$ lies on the grid lines.
(Then, the boundary of $\IPF$ intersects the grid lines transverally. Note however that retrieved points may fall on grid lines.)
We associate every region $R \in F^1$ to the cell containing the top-right corner of $R$.
We do \emph{not} change this association upon retrieving a region.
We denote by $F^{i, j}$ the regions in $F^1$ associated to the cell $(i, j)$.
A cell $(i, j)$ is \emph{non-empty} if $F^{i, j}$ is non-empty.

We further split these regions into three types, depending on the inner Pareto front, see Figure~\ref{fig:classification}.
If $\IPF(F^1)$ does not intersect the cell $(i, j)$, all regions in $F^{i, j}$ are of type A.
Otherwise, since we shifted $F^0$ slightly, the boundary of $\IPF(F^1)$ intersects the boundary of the cell $(i, j)$ at exactly two points.
Any region in $F^{i, j}$ that contains at least one such intersection point is of type B.
All other regions are of type C.
For $k \in \{A, B, C\}$, we denote by $F_k$ the set of regions in $F^1$ of type $k$,
and by $F_k^{i, j}$ the set of regions in $F^{i, j}$ of type $k$.
We also put $F_{AB} = F_A \cup F_B$ 
and $F_{AB}^{i,j} = F_A^{i, j} \cup F_B^{i, j}$,
and use $P_{AB}, P_{AB}^{i, j}$ for the corresponding points in $P$. Note that $F_A^{i, j} \cap F_B^{i, j} = \emptyset$, but we will use $F_{AB}^{i,j}$ to handle regions of type A and~B in the same way.
Finally, we denote $s_k^{i, j} = |F_k^{i, j}|$.

\begin{figure}[b]
    \centering
    \includegraphics{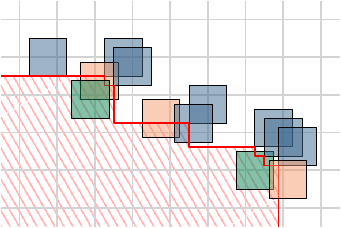}
    \caption{The regions are classified as type A (blue), B (orange), or C (green) depending on their location in relation to the grid and $\IPF(F)$.}
    \label{fig:classification}
\end{figure}

\subparagraph{Preprocessing the grid.}
During preprocessing, we compute $\IPF(F^1)$, overlay the unit grid, and determine the non-empty cells $(i, j)$ and their associated regions.
For each such cell, we classify the regions associated to it using the three types, and store regions of the same type and cell in an array.
We observe that, if $(i, j)$ is a non-empty cell, then at least one of the cells $(i, j), (i-1, j), (i, j-1), (i-1, j-1)$ intersects the boundary of $\IPF(F^1)$. Thus, the non-empty cells lie in a 2-cell wide band along $\IPF(F)$.
We store the non-empty cells in an array, sorted by $j-i$, which allows us to navigate between adjacent non-empty cells in constant time.
In a separate array, we store the non-empty cells sorted by $i+j$. This allows us to sweep these cells from top-right to bottom-left.

\subparagraph{Implementing Stage 1.}
We implement Stage 1 of \cref{alg:strategy_unit_square} by first computing the Pareto front of only the regions of type A and B (Section~\ref{sec:AB}),
and only then considering regions of type C (Section~\ref{sec:C}). This is justified by the following observation

\begin{observation} \label{obs:C_below_AB}
    A region $R$ of type A or B can never lie strictly below the staircase of $\OPF(F_{C})$.
    In particular, $R \in F_{AB}$ lies on $\OPF(F)$ if and only if it lies on $\OPF(F_{AB})$.
\end{observation}

\subsection{Regions of type A and B}\label{sec:AB}
For regions of type A and B, we process the grid cells from top-right to bottom-left,
and apply a modified version of \cref{alg:slow_sweep_line} to each cell $(i, j)$,
with the goal of computing the section of the Pareto front that lies inside this cell.
Observe that, since the regions are unit squares,
a point in $P$ that lies in the cell $(i, j)$ has to stem from a region in $F^{i, j}$, $F^{i+1, j}$, $F^{i, j+1}$, or $F^{i+1, j+1}$.
Among the points outside the cell $(i, j)$ that dominate some part of this cell,
we only care about the rightmost point above $(i, j)$, and the topmost point right of $(i, j)$.
Thus, we apply the sweep line algorithm to these regions and points, see \cref{alg:cell_sweep_line}. 

\begin{algorithm}[tb]
    \caption{The cell-based sweep line algorithm.}
    \label{alg:cell_sweep_line}
    \begin{algorithmic}
        \FORALL{cells $(i, j)$ in top-right to bottom-left order}
            \STATE $p_{\ell} \gets $ rightmost point among $\{p \in P_{AB} \mid p.x \ge i \land p.y \ge j+1\}$
            \STATE $p_r \gets $ topmost point among $\{p \in P_{AB} \mid p.x \ge i+1 \land p.y \ge j\}$
            \STATE $S \gets F_{AB}^{i, j} \cup \{p_\ell, p_r\} \cup \{p \in P_{AB}^{i+1, j} \cup P_{AB}^{i, j+1} \cup P_{AB}^{i+1, j+1} \mid p \text{ was retrieved}\} $
            \STATE Apply \cref{alg:slow_sweep_line} to $S$.
            \STATE Store the part of the Pareto front that lies inside the cell $(i, j)$.
        \ENDFOR
        \STATE To get $\PF(P^1)$, concatenate the stored Pareto fronts cell-by-cell.
    \end{algorithmic}
\end{algorithm}

\begin{theorem} \label{the:cell_sweep_correct}
    \cref{alg:cell_sweep_line} retrieves only regions of type A, B that lie on $\OPF(F)$,
    and only terminates once no such regions exist.
\end{theorem}
\begin{proof}
    Every region $R \in F_{AB}^1$ participates in the execution of \cref{alg:slow_sweep_line} in the cell associated to $R$.
    There, the region $R$ either gets retrieved, or it is dominated by some point~$p^*$ and can thus never lie on $\OPF(F)$.
    Thus, when \cref{alg:cell_sweep_line} terminates, there are no non-point regions of type A, B that lie on $\OPF(F)$.

    Suppose for contradiction that, when processing cell $(i, j)$, \cref{alg:slow_sweep_line} retrieves a region $R$
    that does not lie on $\OPF(F)$. Then, there is a region $R'$ on $\OPF(F)$ and a point $p' \in R'$ such that $p'$ dominates $R$.
    By correctness of \cref{alg:slow_sweep_line}, we have $R' \notin S$.
    Hence, $R'$ is associated to a cell $(i', j') \ne (i, j)$. Since we already processed $(i', j')$, the region $R'$ is a point.
    If $R'$ lies on or above $y=j+1$, then $p_\ell$ dominates $R$.
    Otherwise, $R'$ lies on or to the right of $x=i+1$, and $p_r$ dominates $R$.
    In either case, we get a contradiction which shows that \cref{alg:cell_sweep_line}  retrieves only regions of type A or B that lie on $\OPF(F)$.
\end{proof}

To implement \cref{alg:cell_sweep_line} efficiently,
we should only considers cells $(i, j)$ for which $F_{AB}^{i, j} \cup F_{AB}^{i+1, j} \cup F_{AB}^{i, j+1} \cup F_{AB}^{i+1, j+1}$ is non-empty. There are $O(|F_{AB}|)$ such cells.
These cells lie in a strip around $\IPF(F^1)$ that is at most 3 cells wide,
and we can propagate $p_\ell, p_r$ along this strip in $O(|F^1|)$ time in total.

\begin{theorem} \label{the:cell_running_time}
    Put $t^{i, j} = s_{AB}^{i, j} + s_{AB}^{i+1, j} + s_{AB}^{i, j+1} + s_{AB}^{i+1, j+1}$.
    The running time of \cref{alg:cell_sweep_line} is $ O\Big(|F^1| + \sum_{(i, j) \in \mathbb{Z}^2} t^{i, j} \log(2+t^{i, j})\Big)$.
\end{theorem}
\begin{proof}
    When processing the cell $(i,j)$, we have $|S| \le t^{i, j}$, so constructing the set $S$
    and running \cref{alg:slow_sweep_line} takes $O(t^{i, j} \log (2+t^{i, j}))$ time.
    Updating $p_{\ell}, p_r$ with the points in $P_{AB}^{i, j}$ takes $O(t^{i, j})$ time.
    Propagating $p_{\ell}, p_r$ takes $O(|F^1|)$ time.
\end{proof}

We now show a matching universal lower bound.

\begin{lemma} \label{lem:lower_bound_tool}
    Let $Q$ be a set of points that lie above $\IPF(F^1)$ such that no point in $Q$ dominates another point in $Q$.
    For each region $R \in \hF$, pick at most one point in $Q$ that lies in the interior of $R$.
    Let $F_q$ denote the set of regions that pick a particular point $q \in Q$.
    Then there are at least
    \[
    O_Q := \prod_{q \in Q} |F_q|!
    \]
    possible outputs. Moreover, there is a universal lower bound of
    \[
    \Omega\Big(\sum_{q \in Q} |F_q| \log(2+|F_q|)\Big).
    \]
\end{lemma}
\begin{proof}
    For any region that didn't pick a point in $Q$, place its point on the bottom-left corner of the region. These points cannot dominate any points in $Q$, since all points in $Q$ are above the inner Pareto front.
    For a single point $q \in Q$, consider the regions in $F_q$ and place their points near $q$, but in the top-right quadrant of $q$. This way, these points can form at least $|F_q|!$ different Pareto fronts.
    Moreover, since no point in $Q$ strictly dominates another point in $Q$,
    we can do this placement for all point $q$ simultaneously, to get $O_Q$ different possible outputs.
    By an information theoretic lower bound, we get a universal lower bound of $\Omega (\sum_{q \in Q} \log(|F_q|!) )$.
    \cref{lem:linear_lower_bound} implies an universal lower bound of $\Omega(\sum_{q \in Q} |F_q|)$, and
    \[
    \sum_{q \in Q} \log(|F_q|!) +  \sum_{q \in Q} |F_q| \in \Omega\Big(\sum_{q \in Q} |F_q| \log(2+|F_q|)\Big).\qedhere
    \]
\end{proof}

\begin{lemma} \label{lem:universal_bound_A}
    There is a universal lower bound of $\displaystyle\Omega\Big(\sum_{(i, j) \in \mathbb{Z}^2} s_A^{i, j} \log(2+s_A^{i, j})\Big).$
\end{lemma}
\begin{proof}
    Let $Q$ be set of points $(i, j)$ such that there are regions of type A associated to the cell $(i, j)$. The result follows from \cref{lem:lower_bound_tool}.
\end{proof}

\begin{lemma} \label{lem:universal_bound_B}
    There is a universal lower bound of $ \displaystyle\Omega\Big(\sum_{(i, j) \in \mathbb{Z}^2} s_B^{i, j} \log(2+s_B^{i, j})\Big). $
\end{lemma}
\begin{proof}
    Let $Q$ be the set of intersection points between $\IPF(F^1)$ and the unit grid.
    The result follows from \cref{lem:lower_bound_tool},
    where every cell of type B pick an intersection point of the cell it is associated to.
    Since there are two possible intersection points per cell,
    one of these intersection points get picked by at least $\lceil s_B^{i, j}/2\rceil$ regions.
\end{proof}

\begin{theorem} \label{the:grid_lower_bound}
    Put $t^{i, j} = s_A^{i, j} + s_{AB}^{i+1, j} + s_{AB}^{i, j+1} + s_{AB}^{i+1, j+1}$.
    There is a universal lower bound of $ \Omega\Big(\sum_{(i, j) \in \mathbb{Z}^2} t^{i, j} \log(2+t^{i, j})\Big). $
    In particular, \cref{alg:cell_sweep_line} runs in universally optimal time.
\end{theorem}
\begin{proof}
    To get the lower bound, combine \cref{lem:universal_bound_A,lem:universal_bound_B} with the fact that $x \log (2+x) + y \log (2+y) \ge \frac{1}{4}(x+y) \log (2+x+y)$ for $x, y \ge 0$. Together with \cref{the:cell_running_time,lem:linear_lower_bound}, this shows the universal optimality of \cref{alg:cell_sweep_line}.
\end{proof}

\subsection{Regions of type C}\label{sec:C}

To deal with regions of type C, we no longer use the unit grid.
Instead we run a variant of \cref{alg:slow_sweep_line} with a finger search tree~\cite{brodal2018finger}
and derive a universal lower bound from the order in which the regions and points are considered.
If a point $p$ is a concave corner point of $\IPF(F^1)$ and $p$ lies both on the right facet of one region, and on the top facet of another region, then there is too little freedom to place points near $p$ that appear on the Pareto front, which complicates a lower bound construction.
Instead, we avoid this case by not inserting such points into the finger search tree.

\begin{definition}
    A point $p \sim R$ is \emph{fiddly} if $p$ is a concave corner point of $\IPF(F^1)$
    and $p$ lies on the top facet of $R$.
\end{definition}

If we are processing a region $R$ and the retrieval of $p \sim R$ returns a fiddly point, then, among points above $\IPF(F^1)$, $p$ is the successor of $R$ in sweep order (\cref{def:sweep_order}).
Thus, we can immediately process $p$, without doing a finger insertion. 
See \cref{alg:finger_sweep_line}.

\begin{algorithm}[tb]
    \caption{Finger sweep line algorithm.}
    \label{alg:finger_sweep_line}
    \begin{algorithmic}
        \STATE $p^* \gets (-\infty, \infty)$
        \STATE $S \gets$ points of $P_{AB}$ on $\PF(P_{AB})$, in sweep order
        \STATE $T \gets$ finger search tree of $F_{C}$, sorted by sweep order
        \STATE $L \gets$ empty linked list
        \FOR{ $R \in S \cup T$ in sweep order }
            \IF{$R$ is not dominated by $p^*$}
                \IF{$R$ is a point}
                    \STATE $p^* \gets R$
                    \STATE Append $p^*$ to $L$
                \ELSE
                    \STATE $p \gets$ retrieve $R$.
                    \IF {$p$ is not a fiddly point}
                        \STATE Insert $p$ into $T$ using $R$ as a finger.
                    \ELSIF{$p$ is not dominated by $p^*$}
                        \STATE $p^* \gets p$
                        \STATE Append $p^*$ to L
                    \ENDIF
                \ENDIF
            \ENDIF
        \ENDFOR
        \RETURN $L$
    \end{algorithmic}
\end{algorithm}

\subparagraph{Running time.}
The running time of \cref{alg:finger_sweep_line} is $O(|F^1|)$ plus the time taken for finger insertions.
To analyze these, consider the content of $T$ at the end of the algorithm.
Put $m = |T|$.
Let $R_1, \dots, R_k \in F_C$ be the regions for which we do a finger insertion, with corresponding points $p_1, \dots, p_k \in P$.
Let $a_i$ be the rank of $R_i$ in $T$ and let $b_i$ be the rank of $p_i$ in $T$.
Then $1 \le a_i < b_i \le m$ for all $i \in [k]$.
The finger searches take $O(k + \sum_{i=1}^{k} \log(b_i - a_i))$ time in total.
The following lemma of~\cite{cardinal_sorting_2013} bounds this sum:

\begin{restatable}[A generalization of Lemma 3.2 in~\cite{cardinal_sorting_2013}]{lemma}{intervallemma} \label{lem:interval}
    Let $\{ [a_i, b_i ] \}$ be a set of $k$ intervals where each interval has size at least $1$,
    and $[a_i, b_i] \subseteq [0, m]$.
    Let $Z$ be the set of linear orders realizable by real numbers $z_i \in [a_i, b_i]$.
    Then $ \sum\limits_{i \in [n]} \log( b_i - a_i) \in O(k \log (\frac{em}{k}) + \log |Z|)$.
\end{restatable}
We defer the proof of \cref{lem:interval} to Appendix~\ref{sec:interval_proof}.

\begin{lemma} \label{lem:ipf_embedding}
    Let $k, R_i, p_i, a_i, b_i$ be defined as above. Then there are distinct points $q_1, \dots, q_{m}$ on the boundary $\partial \IPF(F^1)$ in top-left to bottom-right order,
    such that every region $R_i$ contains the points $q_{a_i}, q_{a_i+1}, \dots, q_{b_i}$.
    Moreover, there is a map $\Psi : [1, m] \to \partial \IPF(F^1)$ that is injective and top-left to bottom-right monotone,
    such that every region $R_i$ contains $\Psi([a_i, b_i])$.
\end{lemma}
\begin{proof}
    Figure~\ref{fig:interval_lemma} illustrates the construction. We use $T[i]$ to denote the point (or top-right corner of the region) at index $i$ in~$T$.
    Let $\varepsilon > 0$ be small enough.
    To define $q_i$, start from the leftmost point $\tilde{q}_i$ on $\partial \IPF(F^1)$ at y-coordinate $T[i].y$.
    Then, place $q_i$ near $\tilde{q_i}$ at a distance of $i \cdot \varepsilon$ along $\partial \IPF(F^1)$ to the bottom-right,
    unless $\tilde{q_i}$ is a concave corner of $\IPF(F^1)$, in which case you place $q_i$ near $\tilde{q_i}$ at a distance of $(m + 1 - i) \cdot \varepsilon$ along $\partial \IPF(F^1)$ to the top-left. Finally, define $\Psi(i) = q_i$ and interpolate to an injective map to $\partial \IPF(F^1)$. 

    \begin{figure}
        \centering
        \includegraphics{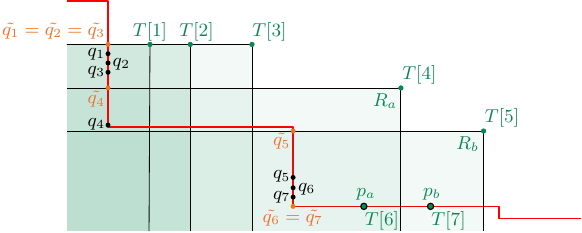}
        \caption{Example where the tree $T$ contains seven regions/points. The point $T[i]$ is the point or top-right corner of the region at index $i$ in $T$, which in turn defines the points $\tilde{q_i}$ and $q_i$. Retrieving the regions  $R_a$ and $R_b$ resulted in the points $p_a = T[6]$ and $p_b = T[7]$ being inserted into $T$. }
        \label{fig:interval_lemma}
    \end{figure}

    Every region $R_i$ contains the points $\tilde{q}_{a_i}, \tilde{q}_{b_i}$, since $R_i$ is a type $C$ region and both points lie to left of $R_i$ and $p_i$.
    Thus, $R_i$ also contains the points in between.
    Since fiddly points do not cause finger insertions, $R_i$ also contains $q_{a_i}, \dots, q_{b_i}$,
    and $\Psi([a_i, b_i])$.
\end{proof}

\begin{theorem} \label{the:cell_sweep_optimal}
    \cref{alg:cell_sweep_line} runs in universally optimal time.
\end{theorem}
\begin{proof}
    Let $a_i, b_i$ be defined as above.
    The running time of \cref{alg:finger_sweep_line} is $O(|F^1| + \sum_{i=1}^{k} \log( b_i - a_i))$.
    By \cref{lem:ipf_embedding}, this is $O(|F^1| + \log |Z|)$ since $k \log(\frac{e m}{k}) \le 2 m \le 2|F^1|$.
    \cref{lem:linear_lower_bound} shows a universal lower bound of $\Omega(|F^1|)$.
    What remains is to show a universal lower bound of $\Omega(\log |Z|)$.
    
    Let $z_i \in [a_i, b_i]$ be distinct real numbers.
    Invoke \cref{lem:ipf_embedding} and put $p_i = \Psi(z_i)$ for $i \in [k]$.
    For all other regions, pick their bottom-left corner as their point.
    This defines a point set $P$ such that $\{p_1, \dots, p_k\}$
    appear on $\PF(P)$ in the linear order formed by the $z_i$.
    Thus, an information theoretic lower bound shows a universal lower bound of $\Omega(\log |Z|)$.
\end{proof}

\subsection{The complete retrieval program for unit squares}

By combining subroutines from the previous section,
we obtain a retrieval program for unit squares in \cref{alg:full_unit_square}.
By \cref{obs:C_below_AB,the:cell_sweep_correct}, this executes \cref{alg:strategy_unit_square}. We thus perform an instance-optimal number of retrievals and correctly compute $\PF(P)$.
The preprocessing takes $O(n \log n)$ time and $O(n)$ space by Lemma~\ref{lem:compute_independent}.
By \cref{lem:linear_lower_bound,the:cell_sweep_optimal,the:grid_lower_bound}, the retrieval phase takes universally optimal time and $O(n)$ space.

\begin{algorithm}[h]
    \caption{Full reconstruction program for unit squares.}
    \label{alg:full_unit_square}
    \begin{algorithmic}
        \STATE Compute $I^0$ and put $F^1 = F^0 - I^0$. \COMMENT{Preprocessing}
        \STATE Sort $I^0$ in top-left to bottom-right order.
        \STATE For every region in $F^1$, compute its predecessor in $I^0$.
        \STATE Overlay the unit grid, compute the non-empty cells $(i, j)$ and store them twice: sorted by $i-j$, and sorted by $i+j$.
        \STATE Compute $\IPF(F^1)$, classify regions into $F_A^{i, j}, F_B^{i, j}, F_C^{i, j}$.
        \STATE Build a finger search tree for $F_C$ in sweep order.
        \STATE Run \cref{alg:cell_sweep_line} to obtain $\PF(P_{AB})$ in sorted order. \COMMENT{Stage 1}
        \STATE Run \cref{alg:finger_sweep_line} to obtain $\PF(P^1)$ in sorted order. 
        \STATE $L \gets $ Merge $I^0$ with $\PF(P^1)$ via the predecessor pointers \COMMENT{Stage 2}
        \FORALL{staircase steps $S$ of $\PF(P^1)$ }
            \STATE Retrieve the regions in $I^0$ that intersect $S$ and remove them from $L$ if dominated.
        \ENDFOR
        \RETURN $L$
    \end{algorithmic}
\end{algorithm}

\begin{theorem}\label{thm:unit_squares}
    Let $F$ be a family of $n$ axis-aligned unit squares. We can process $F$ in $O(n \log n)$ time and $O(n)$ space, such that given $P\sim F$ we reconstruct the Pareto front $\PF(P)$ in universally optimal time.
\end{theorem}

\newpage
\bibliographystyle{plainurl}
\bibliography{references}

\newpage
\appendix

\section{Proof of \cref{lem:interval}} \label{sec:interval_proof}

The interval lemma is a powerful tool introduced by \cite{cardinal_sorting_2013} that keeps on giving \cite{van2019preprocessing,van2024tight,haeupler2025fast,van2025simpler,Haeupler25}:

\begin{lemma} [Lemma 3.2 in~\cite{cardinal_sorting_2013}] \label{lemm:original_interval}
    Let $\{ [a_i, b_i ] \}$ be a set of $n$ intervals where each interval has size at least $1$,
    and $[a_i, b_i] \subseteq [0, n]$.
    Let $Z$ be the set of linear orders realizable by real numbers $z_i \in [a_i, b_i]$.
    Then $ \sum\limits_{i \in [n]} \log( b_i - a_i) \in O(n + \log |Z|)$.
\end{lemma}

We introduce a generalization of this lemma where the number of intervals need not be the same as the size of the universe these intervals live in. This makes the lemma more convenient to use.

\intervallemma*

\begin{proof}
    Many existing proofs of \cref{lemm:original_interval} can be adapted to work in this setting.
    For example, we tweak the proof of \cite{haeupler2025fast}:
    For each $i$ between $1$ and $k$ inclusive, choose a real number $z_i$ uniformly at random from the real interval $[0, m]$, independently for each $i$. With probability $1$, the $r_i$ are distinct. 
    Let $L$ be the permutation of $[k]$ obtained by sorting $[k]$ by $z_i$. Each possible permutation is equally likely.
    If each $z_i \in [a_i, b_i]$ then $L$ is in $Z$.
    The probability of this happening is $\prod_{i=1}^k (b_i-a_i)/m$.  It follows that $|Z| \geq k! \cdot \prod_{i=1}^k (b_i-a_i)/m$.  Taking logarithms gives $\log |Z| \geq \sum_{i=1}^k \log (b_i-a_i) + \log k! - k\log m$.  By Stirling's approximation of the factorial, $\log k! \geq k\log k - k\log e$. Thus,
    \[
    \sum_{i=1}^k \log (b_i-a_i) \le \log |Z| + k \log m - \log k! \le \log |Z| + k \log \Big(\frac{e m}{k}\Big).
    \]
\end{proof}

\subparagraph{Warning.} For $m \le k$, one can show the stronger bound $\sum_{i \in [k]} \log(b_i - a_i) \in O(\log |Z|)$, see Lemma~3 in~\cite{van2025simpler}.
However, this bound no longer true for $m \ge k$. For example, consider $m = 2k$ and the intervals $[0, 2], [2, 4], \dots, [2k-2, 2k]$,
then $|Z| = 1$ but $\sum_{i \in [k]} \log(b_i - a_i) = k \log(2) \in \Theta(k)$.

\end{document}